\documentclass [12pt]{amsart}
\usepackage[utf8]{inputenc}
\pdfoutput=1
\usepackage{amsmath,amssymb,amsthm,amsfonts}
\usepackage{mathtools}%
\usepackage{pgfplots}
\usepackage{comment}
\pgfplotsset{compat=newest}% to avoid the pgfplots warning
\usetikzlibrary{intersections, pgfplots.fillbetween}
\pgfdeclarelayer{pre main}
\usepackage[english]{babel}%
\usepackage{hyperref}%
\usepackage{bm}%
\usepackage{mathrsfs}%

\usepackage{tikz,graphicx,color}
\usepackage{tikz-cd}%
\usepackage{tikz-3dplot}%
\usepackage{tabularx}
\usetikzlibrary{calc}%
\usetikzlibrary{arrows}%
\usetikzlibrary{shapes}%
\usetikzlibrary{patterns}%
\usetikzlibrary{positioning}%
\usetikzlibrary{arrows.meta}
\usetikzlibrary{decorations.markings}
\usetikzlibrary{knots}

\usepackage{epstopdf}%
\usepackage{blkarray}

\usepackage[arrow]{xy}%
\usepackage{diagbox}%
\usepackage{caption}
\usepackage{subcaption}
\usepackage{arcs}%
\usepackage{xcolor}%
\usepackage{xspace}

\usepackage[margin=1.0in]{geometry}%

\usepackage{enumitem}
\usepackage{letltxmacro}
\usepackage{thmtools,etoolbox}

\def\myarabic#1{\normalfont(\roman{#1})}
\newlist{theoremlist}{enumerate}{1}
\setlist[theoremlist]{label=\myarabic{theoremlisti},ref={\myarabic{theoremlisti}},itemindent=0pt,labelindent=0pt,
	leftmargin=*,noitemsep}

\makeatletter
\renewcommand{\p@theoremlisti}{\perh@ps{\thetheorem}}
\protected\def\perh@ps#1#2{\textup{#1#2}}
\newcommand{\itemrefperh@ps}[2]{\textup{#2}}
\newcommand{\itemref}[1]{\begingroup\let\perh@ps\itemrefperh@ps\ref{#1}\endgroup}
\makeatother
\newcommand\restr[2]{{% we make the whole thing an ordinary symbol
  \left.\kern-\nulldelimiterspace % automatically resize the \overline with \right
  #1 % the function
  \vphantom{\big|} % pretend it's a little taller at normal size
  \right|_{#2} % this is the delimiter
  }}

\usepackage{nameref,hyperref}
\usepackage[capitalize]{cleveref}

\newtheorem{theorem}{Theorem}[section]

\newtheorem{lemma}[theorem]{Lemma}
\newtheorem{proposition}[theorem]{Proposition}

\theoremstyle{definition}
\newtheorem{remark}[theorem]{Remark}
\theoremstyle{definition}

\theoremstyle{definition}
\newtheorem{problem}[theorem]{Problem}
\theoremstyle{definition}
\newtheorem{example}[theorem]{Example}

\newcommand{\bp}{\begin{problem}}
\newcommand{\ep}{\end{problem}}
\newcommand{\bs}{\begin{proof}[Solution]}
\newcommand{\es}{\end{proof}}

\newcommand{\Z}{\mathbb{Z}}
\newcommand{\C}{\mathbb{C}}
\newcommand{\R}{\mathbb{R}}
\newcommand{\wt}{\mathrm{wt}}

\newcommand{\w}{\mathrm{w}}

\newcommand{\ra}{\rightarrow}

\usepackage{mathtools}

\crefname{figure}{Figure}{Figures}

\def\figref#1(#2){Figure~\hyperref[#1]{\ref*{#1}(#2)}}

\addtotheorempostheadhook[theorem]{\crefalias{theoremlisti}{theorem}}
\addtotheorempostheadhook[lemma]{\crefalias{theoremlisti}{lemma}}
\addtotheorempostheadhook[proposition]{\crefalias{theoremlisti}{proposition}}
\addtotheorempostheadhook[corollary]{\crefalias{theoremlisti}{corollary}}

\definecolor{calpolypomonagreen}{rgb}{0, 0.6, 0.2}

\newcounter{todofigure}

\tikzset{qvert/.style={draw,black,circle,fill=gray,minimum size=5pt,inner sep=0pt}  } 
\tikzset{bvert/.style={draw,circle,fill=black,minimum size=5pt,inner sep=0pt}  }  
\tikzset{gbvert/.style={draw, gray, circle,fill=gray,minimum size=5pt,inner sep=0pt}  } 
\tikzset{gvert/.style={draw,gray,circle,fill=white,minimum size=5pt,inner sep=0pt}  } 
\tikzset{wvert/.style={draw,circle,fill=white,minimum size=5pt,inner sep=0pt}  } 
\tikzset{fvert/.style={text=MidnightBlue}  } 
\tikzset{sqvert/.style={draw,black,rectangle,fill=black,minimum size=5pt,inner sep=0pt}  } 
\tikzset{lvert/.style={draw,circle,fill=black,minimum size=4pt,inner sep=0pt}  }  
\usetikzlibrary{arrows}
\tikzcdset{arrow style=tikz, diagrams={>={Stealth[round,length=4pt,width=4.95pt,inset=2.75pt]}}}
\tikzset{nvert/.style={draw,circle,fill=black,minimum size=3pt,inner sep=0pt}  }

\def\ep{0.16}
\def\rc{7}
\def\lw{1.1pt}

\def\T{\mathbb T}
\def\Rpos{\mathbb R_{>0}}
\def\b{\mathrm b}
\def\bw{\mathrm w}
\numberwithin{equation}{section}

\def\res{\operatorname{Res}}

\begin{document}
	\numberwithin{equation}{section}
	
	\title{Spectral transform for the Ising model}
	\author{Terrence George}
	\address{Department of Mathematics, University of California, Los Angeles, CA 90095, USA}
	\email{{\href{mailto:tegeorge@math.ucla.edu}{tegeorge@math.ucla.edu}}}
	\date{\today}
	
	\begin{abstract}
		We prove a correspondence between Ising models in a torus and the algebro-geometric data of a Harnack curve with a certain symmetry and a point in the real part of its Prym variety, extending the correspondence between dimer models and Harnack curves and their Jacobians due to Kenyon and Okounkov. 
	\end{abstract}
	
	\maketitle

\section{Introduction}	
	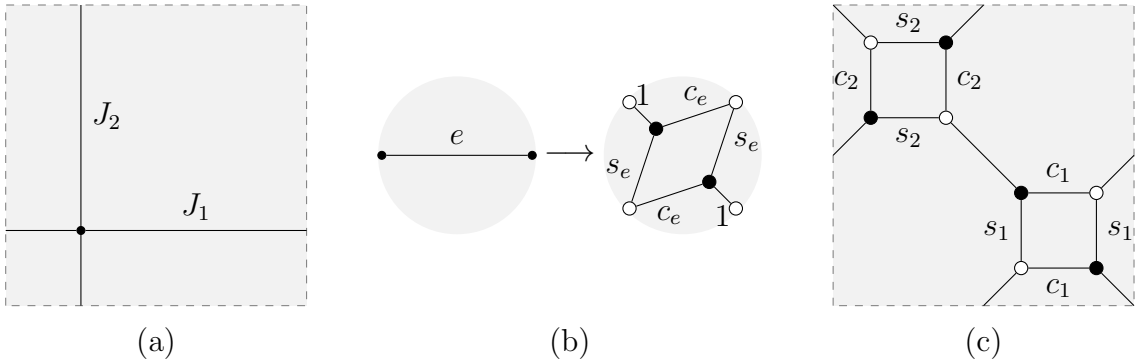
\begin{figure}
		\begin{tikzpicture}[scale=0.5]
			\def\ep{0.3}
			\def\xsh{1}
			% \node[](no) at (0,-3){(M1) The spider move.};
			\begin{scope}[shift={(-\xsh,0)},rotate=0]
	
			  % \draw[gray,dashed,-] (0,0) circle (\r cm);
			  \fill[black!5] (0,0) rectangle (8,8);
			  \draw[gray,dashed] (0,0) rectangle (8,8);
			
			  \coordinate[nvert] (n) at (2,2);
			  \draw[] (n)--node[above]{$J_{1}$} (8,2);
			  \draw[](0,2) --(n);
			  \draw[] (2,0) -- (n)
			  (n) -- node[right]{$J_{2}$} (2,8);

			  \node[](no) at (4,-1) {(a)};
			\end{scope}

			\def\sh{14}
			\def\ysh{4}
			\def\ep{0.3}
			\node[](no) at (\sh,\ysh){$\longrightarrow$};
			% \node[](no) at (0,-3){(M1) The spider move.};
			\begin{scope}[shift={(\sh+3,\ysh)},rotate=90+45]
			  \def\r{2};
			  % \draw[gray,dashed,-] (0,0) circle (\r cm);
			  \fill[black!5] (0,0) circle (1.05*\r cm);
			  \coordinate[wvert] (n1) at (0:\r);
			  \coordinate[wvert] (n2) at (0+90:\r);
			  \coordinate[wvert] (n3) at (0+180:\r);
			  \coordinate[wvert] (n4) at (0+270:\r);
			  
			  \coordinate[bvert] (b1) at (0:0.5*\r);
			  \coordinate[bvert] (b2) at (0+180:0.5*\r);
			  
			  \draw[-]
			  (n1) -- node[above]{$1$} (b1) -- node[left]{$s_e$} (n2)-- node[below]{$c_e$} (b2)-- node[right]{$s_e$}(n4) -- node[above]{$c_e$} (b1)
			  (b2)--node[below]{$1$}(n3)
			  ;

			\end{scope}
			\node[](no) at (14,-1) {(b)};
			
			\begin{scope}[shift={(\sh-3,\ysh)},rotate=0]
			  \def\r{2};
			  
			  \fill[black!5] (0,0) circle (1.05*\r cm);
			  \coordinate[nvert] (n1) at (0:\r);
			  \coordinate[nvert] (n2) at (0+180:\r);  
	\draw (n1) -- node[above]{$e$} (n2);

			\end{scope}

			\begin{scope}[shift={(20+\xsh,0)},rotate=0]

				% \draw[gray,dashed,-] (0,0) circle (\r cm);
				\fill[black!5] (0,0) rectangle (8,8);
				\draw[gray,dashed] (0,0) rectangle (8,8);

				\coordinate[bvert] (b1) at (1,5);
				\coordinate[bvert] (b2) at (1+2,5+2);
				\coordinate[bvert] (b3) at (5,3);
				\coordinate[bvert] (b4) at (7,1);
	  
				\coordinate[wvert] (w1) at (1,7);
				\coordinate[wvert] (w2) at (1+2,5);
				\coordinate[wvert] (w3) at (5,1);
				\coordinate[wvert] (w4) at (7,3);
	  
				\draw[] (b1) -- (0,4) 
				(b1) --node[below]{$s_2$} (w2)
				(b1) --node[left]{$c_2$} (w1);
	  
				\draw[] (b2) --node[above]{$s_2$} (w1)
				(b2) -- (4,8)
				(b2) --node[right]{$c_2$} (w2);
	  
				\draw[] (b3) --node[left]{$s_1$} (w3)
				(b3) -- (w2)
				(b3) --node[above]{$c_1$} (w4) ;
	  
				\draw[] (b4) --node[below]{$c_1$} (w3)
				(b4) -- (8,0)
				(b4) --node[right]{$s_1$} (w4) ;
	  
				\draw[] (w4) -- (8,4)
				(0,8) -- (w1)
				(4,0) -- (w3);
				\node[](no) at (4,-1) {(c)};
			   \end{scope}

		  \end{tikzpicture}
		  \caption{(a) An Ising model $(G,J)$ in a torus $\T$, (b) the mapping from $(G,J)$ to $(G^\square,[\wt^\square])$ and (c) the dimer model $(G^\square,[\wt^\square])$.} \label{fig:isingintro}
	\end{figure}

An \textit{Ising model} in a torus $\T$ is a pair $(G,J)$ where $G=(V,E,F)$ is a graph embedded in $\T$ such that every face of $G$ is a topological disk and $J:E(G) \ra \R_{>0}$ is a function called the \textit{coupling constant}. %Let $\mathcal I_G(\Rpos):=\{J:E(G) \ra \Rpos\}= \Rpos^{E(G)}$ denote the set of Ising models with underlying graph $G$. 
A \textit{dimer model} in $\T$ is a pair $(\Gamma,[\wt])$ where $\Gamma=(B \sqcup W, E,F)$ is a bipartite graph in $\T$ and $[\wt]$ is a function $\wt:E(\Gamma) \ra \Rpos$ called \textit{edge weight} defined modulo a certain gauge equivalence. %Let $\mathcal X_\Gamma(\Rpos)$ denote the space of dimer models with graph $\Gamma$. 

Following \cite{FW,Dubedat,CdT2}, we can associate to an Ising model $(G,J)$ a dimer model $(G^\square,[\wt^\square])$ as follows. We define two functions $s, c : E \ra (0,1)$ by 
\[
s_e := \operatorname{sech}(2J_e) \text{ and }c_e:= \operatorname{tanh}(2J_e) 
\]
for every edge $e \in E$ and replace $e$ with the bipartite graph shown in Figure~\ref{fig:isingintro}(b). %This defines an embedding $\mathcal I_{G}(\Rpos) \subset \mathcal X_{G^\square}(\Rpos)$. 

\begin{figure}
	\begin{tikzpicture}[scale=0.6]
		\def\ep{0.3}
		
		% \node[](no) at (0,-3){(M1) The spider move.};
		\begin{scope}[shift={(4,0)},rotate=45]
			\def\r{2};
			% \draw[gray,dashed,-] (0,0) circle (\r cm);
			\fill[black!5] (0,0) circle (1.05*\r cm);
			\coordinate[wvert] (n1) at (0:\r);
			\coordinate[wvert] (n2) at (0+90:\r);
			\coordinate[wvert] (n3) at (0+180:\r);
			\coordinate[wvert] (n4) at (0+270:\r);
			
			\coordinate[bvert] (b1) at (0:0.5*\r);
			\coordinate[bvert] (b2) at (0+180:0.5*\r);
			
			\draw[-]
			(n1) -- node[above]{$1$} (b1) -- node[above]{$ \frac{d}{ac+bd}$} (n2)-- node[left]{$ \frac{c}{ac+bd}$} (b2)-- node[below]{$\frac{b}{ac+bd}$}(n4) -- node[right]{$\frac{a}{ac+bd}$} (b1)
			(b2)--node[below]{$1$}(n3)
			;
		 
		\end{scope}
		
		\begin{scope}[shift={(-3,0)},rotate=90+45]
			\def\r{2};
			% \draw[gray,dashed,-] (0,0) circle (\r cm);
			\fill[black!5] (0,0) circle (1.05*\r cm);
			\coordinate[wvert] (n1) at (0:\r);
			\coordinate[wvert] (n2) at (0+90:\r);
			\coordinate[wvert] (n3) at (0+180:\r);
			\coordinate[wvert] (n4) at (0+270:\r);
			
			\coordinate[bvert] (b1) at (0:0.5*\r);
			\coordinate[bvert] (b2) at (0+180:0.5*\r);
			
			\draw[-]
			(n1) -- node[above]{$1$} (b1) -- node[left]{$ a$} (n2)-- node[below]{$ d$} (b2)-- node[right]{$ c$}(n4) -- node[above]{$ b$} (b1)
			(b2)--node[below]{$1$}(n3)
			;
			\coordinate[] (t1) at (15:\r);
			\coordinate[] (t2) at (120-45:\r);
			\coordinate[] (t3) at (150-45:\r);
			\coordinate[] (t4) at (210-45:\r);
			\coordinate[] (t5) at (240-45:\r);
			\coordinate[] (t6) at (300-45:\r);
			\coordinate[] (t7) at (330-45:\r);
			\coordinate[] (t8) at (30-45:\r);

		  \end{scope}
	
	\draw[] (7,-2) -- (7,2);
	
	\node[](no) at (0.5,0) {$\longleftrightarrow$};		
	\node[](no) at (0.5,-3) {(a)};	
	\node[](no) at (13+0.5,-3) {(b)};	
			\begin{scope}[shift={(12-2,0)}
				,rotate=-45
				]
				\def\r{2};
				\fill[black!5] (0,0) circle (\r cm);

				\coordinate[] (t1) at (15:\r);
				\coordinate[] (t2) at (120-45:\r);
				\coordinate[] (t3) at (150-45:\r);
				\coordinate[] (t4) at (210-45:\r);
				\coordinate[] (t5) at (240-45:\r);
				\coordinate[] (t6) at (300-45:\r);
				\coordinate[] (t7) at (330-45:\r);
				\coordinate[] (t8) at (30-45:\r);
				
				\coordinate[] (m1) at (0:1);
				\coordinate[] (m2) at (0+90:1);
				\coordinate[] (m3) at (0+180:1);
				\coordinate[] (m4) at (0+270:1);

				\coordinate[wvert] (n1) at (0:\r);
				\coordinate[wvert] (n2) at (0+90:\r);
				\coordinate[wvert] (n3) at (0+180:\r);
				\coordinate[wvert] (n4) at (0+270:\r);

				\coordinate[bvert] (n5) at (0,0);
				\draw[] (n5) edge node[above]{$b$} (n2) edge node[above]{$a$} (n3) edge node[below]{$d$}  (n4) edge node[below]{$c$}  (n1);
				\end{scope}
				\node[](no) at (13-0.5+1,0) {$\longleftrightarrow$};
					\begin{scope}[shift={(20-2-1,0)},rotate=-45]
				\def\r{2};
				\fill[black!5] (0,0) circle (\r cm);
				\coordinate[wvert] (n1) at (0:\r);
				\coordinate[wvert] (n2) at (0+90:\r);
				\coordinate[wvert] (n3) at (0+180:\r);
				\coordinate[wvert] (n4) at (0+270:\r);
				\coordinate[wvert] (n5) at (0,0);
				
				\coordinate[bvert] (b1) at (-0.5,0.5);
				\coordinate[bvert] (b2) at (0.5,-0.5);
				
				\draw[-]
								(n1)--node[below]{$c$}(b2)--node[below]{$d$}(n4)
								(n2)--node[above]{$b$}(b1)--node[above]{$a$}(n3)
								(b1)--node[left]{$1$}(n5)--node[left]{$1$}(b2)
								;

				\end{scope}  
	  \end{tikzpicture}
	  \caption{(a) The square move and (b) the contraction-uncontraction move. Using gauge equivalence, we can assume that the original weight is as shown on the right. Then the new weight is as shown on the left.} \label{fig:squareintro}
\end{figure}
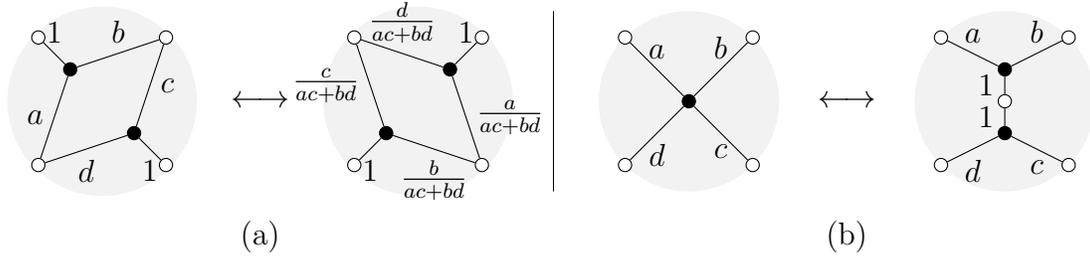

There are two local moves on bipartite torus graphs (Figure~\ref{fig:squareintro}). If we apply square moves at all the square faces of $G^\square$ corresponding to edges of $G$, the resulting graph $\overline{G^\square}$ is the same as $G^\square$ but with all vertices having opposite colors. Let $\mu([\wt])$ denote the resulting weight on $\overline{G^\square}$. Let $[\overline \wt]$ denote the weight where the weight of every edge is the same as that of $G^\square$. 

Our first result is the following characterization of the subset of dimer models that corresponds to Ising models.

\begin{theorem}
A dimer model $(G^\square,[\wt])$ arises from an Ising model if and only if $[\overline \wt] = \mu([\wt])$.
\end{theorem}

One direction is easy. If $\wt^\square$ is as in Figure~\ref{fig:isingintro}(b) then $a=c=s_e$ and $b=d=c_e$ in Figure~\ref{fig:squareintro}(a). Since $ac+bd = s_e^2+c_e^2=1$, we get $[\overline{ \wt^\square}] = \mu([\wt^\square])$. To prove the other direction, we have to study the conditions imposed by setting $[\overline \wt] = \mu([\wt])$ on the gauge invariant $X$ coordinates on the set of dimer models. This is done in Section~\ref{sec:car_weight}.

Kenyon and Okounkov \cite{KO} showed that for every (minimal) bipartite graph $\Gamma$ in $\T$ there is a bijection called the \textit{spectral transform} between the following sets of data.
\begin{enumerate}
	\item[\bf{Data A.}] 
	\begin{enumerate}
		\item $(\Gamma,[\wt])$ a dimer model in $\T$.
		\item $[\kappa]$ a Kasteleyn sign on $\Gamma$ (equivalent to a choice of one of the four elements of $H^1(\T,\{\pm 1\})$). 
	\end{enumerate}
	\item[\bf{Data B.}] \begin{enumerate}
		\item $C \subset (\C^\times)^2$ a \textit{spectral curve} which is a real algebraic curve of a special type called a Harnack curve.
		\item $D$ a divisor on $C$ with one point on each compact oval of $C$ (i.e., connected component of the real locus $C(\R)$) called a \textit{standard divisor}. Such divisors form a component of the real part of the Jacobian variety $\operatorname{Jac}(C)$ of $C$.
		\item An ordering of the points at toric infinity of $C$.
	\end{enumerate}  
\end{enumerate}

Our second result is the following proved in Section~\ref{sec:isingspecdata}.

\begin{theorem}
	For every (minimal) graph $G$ in $\T$, the spectral transform restricts to a bijection between the following sets of data.
\begin{enumerate}
	\item[\bf{Data A.}] 
	\begin{enumerate}
		\item $(G,J)$ an Ising model in $\T$.
		\item $[\kappa]$ a Kasteleyn sign on $G^\square$. 
	\end{enumerate}
	\item[\bf{Data B.}] \begin{enumerate}
		\item $C$ a Harnack curve that is invariant under the involution $\sigma$ on $(\C^\times)^2$ defined by $(z,w) \mapsto (z^{-1},w^{-1})$. Whenever a curve $C$ carries an involution, it defines a linear subvariety of $\operatorname{Jac}(C)$ called the \textit{Prym variety} of $C$.
		\item $D$ a standard divisor in the Prym variety of $C$.	
		\item An ordering of the points at toric infinity of $C$ that is invariant under $\sigma$.
	\end{enumerate}  
\end{enumerate}
\end{theorem}

This correspondence is interesting for several reasons. Firstly, it generalizes the $Z$-invariant Ising models of Baxter~\cite{Baxter1,Baxter2,Baxter3}, which correspond to genus-one spectral curves,	 to arbitrary genus. These Ising models and their genus-zero degenerations (critical or isoradial Ising models) have been extensively studied \cite{CS1,BdT,CS2,Cim,CDC,ZLi,BdTR,Gal}.

Secondly, it completes the following table of correspondences between statistical-mechanical models and algebro-geometric data.

\begin{center}
	\begin{tabularx}{0.9\textwidth} { 
			 |>{\centering\arraybackslash}X 
			 >{\centering\arraybackslash}X 
			 >{\centering\arraybackslash}X|  }
		\hline
		& disk & torus \\
		\hline
		dimer models  & positive Grassmannian \cite{Post}& Harnack curves and standard divisors \cite{KO} \\
		\hline
		electrical networks  & positive Lagrangian Grassmannian \cite{Lam,BGKT,CGS} &  symmetric Harnack curves with a node and standard divisors in Prym varieties \cite{G1} \\
		\hline
		Ising models  & positive orthogonal Grassmannian \cite{HW,Galpy}  & symmetric Harnack curves and standard divisors in Prym varieties (this paper) \\
		\hline
	\end{tabularx}
\end{center}
Thirdly, combining results of Fock \cite{Fock} with compatibility of the Ising Y-$\Delta$ move with the local moves on bipartite graphs \cite{KP}, we get that the discrete dynamical systems on Ising models arising from the Ising Y-$\Delta$ move are linearized on Prym varieties of spectral curves; in this sense, Ising models in $\T$ give rise to integrable systems extending the cluster integrable systems of \cite{GK} constructed from dimer models in $\T$. 

Finally, the spectral transform for the dimer model has been recently used in statistical mechanics to understand limit shapes \cite{BDuits, BB}. We expect that the spectral transform for the Ising model will have similar applications in statistical mechanics.

\subsection*{Acknowledgements.} I first started thinking of the spectral transform of the Ising model in 2019 when I was a graduate student after discussions with Richard Kenyon. I am grateful to C\'{e}dric Boutillier, B\'{e}atrice de Tili\`ere and Pavel Galashin for discussions. 

\section{Background on the Ising model}
In this section, we collect some background on the Ising model. For further background, see \cite{Csur,CCK}.
Let $G=(V,E,F)$ be a graph in $\T$ and let $\mathcal I_G(\Rpos):=\{J:E(G) \ra \Rpos\}$ be the set of Ising models with graph $G$.

\subsection{Zig-zag paths and the Newton polygon}
%Two Ising models $(G,J)$ and $(G',J')$ are said to be \textit{move equivalent} if we can transform $G$ into $G'$ using Y-$\Delta$ moves and duality. Move equivalence classes of Ising models are classified by an invariant called the \textit{Newton polygon}. 
A \textit{zig-zag path} in $G$ is an oriented path in $G$ that alternately turns maximally left or right at each vertex. Zig-zag paths in $G$ come in pairs with opposite orientations; we denote by $\overline \alpha$ the zig-zag path opposite to the zig-zag path $\alpha$. Let $Z_G$ denote the set of zig-zag paths in $G$. 

%The \textit{medial graph} of $G$ is the graph $G^\times$ defined as follows:
%\begin{enumerate}
%\item There is a vertex $t_e$ of $G^\times$ for each edge $e \in E(G)$, which we place at the midpoint of $e$.
%\item  For $e_1,e_2 \in E(G)$, we connect $t_{e_1}$ and $t_{e_2}$ by an edge in $G^\times$ if $e_1$ and $e_2$ are consecutive edges around a face of $G$.
%\end{enumerate}
%It is customary to identify the zig-zag path $e_1 \ra e_2 \ra \cdots \ra e_n \ra e_1$ with the path $t_{e_1} \rightarrow t_{e_2} \rightarrow \cdots \rightarrow t_{e_n}\rightarrow t_{e_1}$ in the medial graph.

Let $\pi: \R^2 \ra \T$ denote the universal cover of $\T$. We say that $G$ is \textit{minimal} if the lift of any zig-zag path to $\R^2$ does not have a self-intersection and the lifts of two zig-zag paths to $\R^2$ share at most one edge. Hereafter, we assume that our graphs are minimal.

Choose a fundamental rectangle $R$ for $\T$ and let $\gamma_z,\gamma_w$ be loops in $\T$ along the sides of $R$ as shown in Figure~\ref{fig:ising_torus}(a). Then $\{[\gamma_z], [\gamma_w]\}$ is a basis for $H_1(\T,\Z)$ identifying $H_1(\T,\Z)$ with $\Z^2$.
A nonzero vector $v \in \Z^2$ is called \textit{primitive} if $v$ is not a multiple of another vector in $\Z^2$, i.e., if $v = \lambda w$ for $ w \in \Z^2$ and $\lambda \in \Z_{>0}$ then $v=w$ and $\lambda = 1$. A convex polygon $N \subset \R^2$ is said to be \textit{integral} if all of its vertices are contained in $\Z^2$. Let $\sigma$ denote the involution $v \mapsto -v$ of $\R^2$. A convex integral polygon $N$ is called \textit{centrally symmetric} if $N$ is invariant under $\sigma$. 

Associated to a minimal $G$ is a centrally-symmetric convex integral polygon $N_G \subset \R^2$ as follows: each zig-zag path $\alpha \in Z_G$ defines a primitive vector given by the homology class $[\alpha] \in H_1(\T,\Z) \cong \Z^2$. There is a unique (modulo translation) convex integral polygon in $\R^2$ whose counterclockwise-oriented boundary consists of the vectors $\{[\alpha] \in \Z^2: \alpha \in Z_G\}$. The translation is fixed by centering at $(0,0)$.

\begin{example}
	Let $(G,J)$ denote the Ising model in Figure~\ref{fig:isingintro}(a). There are four zig-zag paths in $G$. Let $\alpha$ (resp., $\beta$) denote the red (resp., blue) zig-zag path in Figure~\ref{fig:ising_torus}(a). The other two zig-zag paths are $\overline \alpha$ and $\overline \beta$. The homology classes are 
	\[
	[\alpha] = (1,1),\quad [\beta] = (-1,1),\quad [\overline \alpha]=(-1,-1),\quad [\overline \beta] = (1,-1),	
	\]
	so the Newton polygon is the convex integral polygon shown in Figure~\ref{fig:ising_torus}(b). 
\end{example}

\begin{comment}
\begin{theorem}[\cite{GK}] 
	For every centrally symmetric convex integral polygon $N$, there is a family of minimal graphs with Newton polygon $N$. Any two members of a family are move equivalent. In other words, the set of move equivalence classes of minimal graphs is in bijection with the set of centrally symmetric convex integral polygons.
\end{theorem}

Given a centrally symmetric convex integral polygon $N$, we define the \textit{set of positive points of the Ising cluster variety} to be 
\[
	\mathcal I_N(\Rpos):=\left(\bigsqcup_{\text{minimal graphs }G: N_G=N} \mathcal I_G(\Rpos)\right) \text{ modulo } (\text{Y-$\Delta$ moves and duality}).
\]
\end{comment}

\begin{figure}
	\begin{tikzpicture}[scale=0.5]
		\def\ep{0.3}
		% \node[](no) at (0,-3){(M1) The spider move.};

		\begin{scope}[shift={(10,0)},rotate=0]

			% \draw[gray,dashed,-] (0,0) circle (\r cm);
			\fill[black!5] (0,0) rectangle (8,8);
			\draw[gray,dashed] (0,0) rectangle (8,8);
		  
			\coordinate[nvert] (n) at (2,2);
			\draw[] (n)--(8,2);
			\draw[](0,2) --(n);
			\draw[] (2,0) -- (n)
			(n) -- (2,8);

			\draw[red,->,line width=\lw,rounded corners=\rc] (2+\ep,0) -- (2+\ep,2-\ep) -- (8,2+\ep);
			\draw[red,->,line width=\lw,rounded corners=\rc] (0,2+\ep) -- (2-\ep,2+\ep) -- (2+\ep,8);

			\draw[blue,->,line width=\lw,rounded corners=\rc]  (8,2-\ep) -- (2+\ep,2+\ep) -- (2-\ep,8);
			\draw[blue,->,line width=\lw,rounded corners=\rc]  (2-\ep,0) -- (2-\ep,2-\ep) -- (0,2-\ep);

			\draw[->,densely dotted] (7.5,0) -- (7.5,8);
			\draw[->,densely dotted] (0,7.5) -- (8,7.5);
			\node (no) at (-0.5,7.5) {$\gamma_z$};
			\node (no) at (7.5,-0.5) {$\gamma_w$};
		  \end{scope}
\node[](no) at (14,-1) {(a)};

		  \begin{scope}[shift={(22,4)},rotate=0,scale=2]
			\coordinate[nvert] (00) at (0,0);
			\coordinate[nvert] (10) at (1,0);
			\coordinate[nvert] (01) at (0,1);
			\coordinate[nvert] (m10) at (-1,0);
			\coordinate[nvert] (0m1) at (0,-1);
\draw[red,->, line width = \lw] (0m1)--(10);
\draw[red,->, line width = \lw] (01)--(m10);
\draw[blue,->, line width = \lw] (10)--(01);
\draw[blue,->, line width = \lw] (m10)--(0m1);

\node[](no) at (0,-2.5) {(b)};
		  \end{scope}
	
	  \end{tikzpicture}
	  \caption{(a) Two of the zig-zag paths of the Ising model $(G,J)$ in Figure~\ref{fig:isingintro}(a) and (b) its Newton polygon. } \label{fig:ising_torus}
\end{figure}
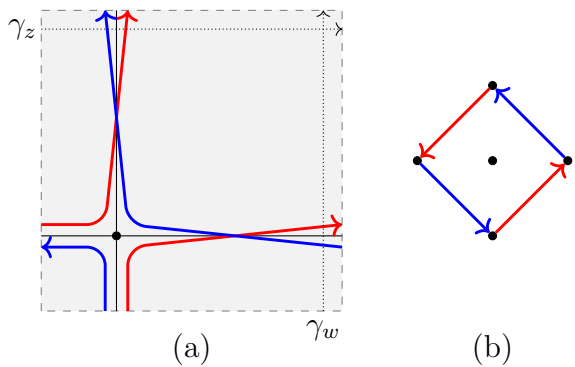	

\subsection{Kramers--Wannier duality}
Let $(G,J)$ be an Ising model in $\T$. Let $G^*$ denote the dual graph of $G$ in $\T$. For $e \in E(G)$, let $e^* \in E(G^*)$ denote the dual edge. Define $J^*:E(G^*) \ra \Rpos$ by the condition
\[
\operatorname{sinh}(2J_{e^*}^*) = \frac{1}{ \operatorname{sinh}(2J_e)}. 	
\]
Then $(G^*,J^*)$ is called the \textit{dual Ising model}. Duality defines a bijection 
\[
\mathcal I_G(\Rpos) \ra \mathcal I_{G^*}(\Rpos).	
\]

It is convenient to introduce the following coordinates. Let $x:E(G) \ra \Rpos$ be defined as $x_e := \exp(2J_e)$. Then $x^*_{e^*}:=\exp(2J^*_{e^*})$ is given by the unique positive solution to 
\[
	x_e+x^*_{e^*}+x_ex^*_{e^*}=1.
\]

\subsection{Y-$\Delta$ move}

\begin{figure}
	\begin{tikzpicture}[scale=0.6]
		\def\ep{0.3}

		\begin{scope}[shift={(-3,0)}, rotate=180]
			\def\r{2};
			% \draw[gray,dashed,-] (0,0) circle (\r cm);
			\fill[black!5] (0,0) circle (1.05*\r cm);
			\coordinate[nvert] (n1) at (0,0);
			\coordinate[nvert] (n2) at (90:\r);
			\coordinate[nvert] (n3) at (90+120:\r);
			\coordinate[nvert] (n4) at (-30:\r);

			\draw[-]
			(n1) edge node[left]{$a$} (n2) edge node[above]{$c$}(n3) edge node[above]{$b$}(n4)
			;

		  \end{scope}

		\node[](no) at (0.5,0) {$\longleftrightarrow$};	
		
		\begin{scope}[shift={(4,0)}, rotate=180]
			\def\r{2};
			% \draw[gray,dashed,-] (0,0) circle (\r cm);
			\fill[black!5] (0,0) circle (1.05*\r cm);
			
			\coordinate[nvert] (n2) at (90:\r);
			\coordinate[nvert] (n3) at (90+120:\r);
			\coordinate[nvert] (n4) at (-30:\r);

			\draw[-]
			(n2) -- node[right]{$B$} (n3) -- node[above]{$A$}(n4) -- node[left]{$C$}(n2)
			;

		  \end{scope}
	
	  \end{tikzpicture}
	  \caption{The Ising Y-$\Delta$ move.} \label{figyd}
\end{figure}
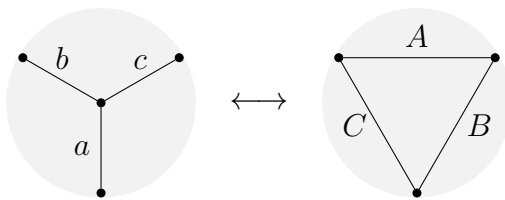

Ising models have a local transformation called the Y-$\Delta$ or star-triangle move which replaces a portion of a graph $G$ that looks like one side of Figure~\ref{figyd} with a portion that looks like the other side to get a graph $G'$ and modifies the $x$ weights by
\begin{align*}
	A=\sqrt{ \frac{(abc+1)(a+bc)}{(b+ac)(c+ab)}},\quad
	B=\sqrt{ \frac{(abc+1)(b+ac)}{(a+bc)(c+ab)}},\quad
	C=\sqrt{ \frac{(abc+1)(c+ab)}{(a+bc)(b+ac)}}.
\end{align*}
The Y-$\Delta$ move gives rise to a bijection $\mathcal I_{G}(\Rpos) \ra \mathcal I_{G'}(\Rpos)$. Moreover, the Y-$\Delta$ move and duality are compatible, i.e., the following diagram commutes
\[
\begin{tikzcd}
	\mathcal I_{G}(\Rpos) \arrow[r,"\text{Y-$\Delta$ move}"]\arrow[d,"\text{duality}"] &\mathcal I_{G'}(\Rpos) \arrow[d,"\text{duality}"]\\
	\mathcal I_{G^*}(\Rpos) \arrow[r,"\text{Y-$\Delta$ move}"] & \mathcal I_{(G')^*}(\Rpos)
\end{tikzcd}.
\]

We say that two graphs $G$ and $G'$ are \textit{move-equivalent} if they are related a sequence of Y-$\Delta$ moves and duality. 
\begin{theorem}[{\cite[Theorem 5.4]{GK}}] 
	For every centrally symmetric convex integral polygon $N$, there is a family of minimal graphs with Newton polygon $N$. Any two members of a family are move equivalent. In other words, the set of move equivalence classes of minimal graphs is in bijection with the set of centrally symmetric convex integral polygons.
\end{theorem}
\section{Background on the dimer model}
In this section, we give a brief background on the dimer model and its spectral transform, mostly following~\cite{Ken, GK}.
\subsection{Dimer models in $\T$}

A \textit{dimer model} in $\T$ is a pair $(\Gamma,[\wt])$ where 
\begin{enumerate}
	\item $\Gamma=(B \sqcup W,E,F)$ is a bipartite graph embedded in $\T$ such that the faces of $\Gamma$ are topological disks,
	\item $[\wt]$ is the gauge-equivalence class of $\wt:E  \ra \C^\times$ which is a function assigning to each edge its \textit{edge weight}, and two edge weights $\wt_1$ and $\wt_2$ are \textit{gauge equivalent} if there is a function $f:B \sqcup W \ra \C^\times$ such that for every edge $e=\{\b,\bw\} \in E$ (where $\b \in B, \bw \in W$),
	\[
		\wt'(e) = f(\b)^{-1} \wt(e) f(\bw).
		\]
		
\end{enumerate}

	Let $\mathcal X_\Gamma$ denote the space of edge weights modulo gauge equivalence. It will be convenient to rephrase the above in the language of algebraic topology. In particular, doing so will help us identify coordinates on $\mathcal X_\Gamma$ that are invariant under gauge equivalence.

 We consider $\Gamma$ as a cell complex where the vertices are the $0$-cells and the edges are the $1$-cells, where we orient each edge $e = \{\b,\bw\}$ from $\b$ to $\bw$. The cellular chain complex is
\[
0 \ra C_1(\Gamma,\Z) \xrightarrow[]{\partial} C_0(\Gamma,\Z)\ra 0,	
\]
where $C_0(\Gamma,\Z) = \Z B \oplus \Z W$, $C_1(\Gamma,\Z) = \Z E$ and $\partial e = \bw - \b$. Therefore, 
\[
H_1(\Gamma,\Z) = \ker (C_1(\Gamma,\Z) \xrightarrow[]{\partial} C_0(\Gamma,\Z)),	
\]
so $1$-homology classes in $\Gamma$ are the same thing as $1$-cycles in $\Gamma$. 

The cellular cochain complex is 
\[
	1 \ra C^0(\Gamma,\C^\times) \xrightarrow[]{d} C^1(\Gamma,\C^\times)\ra 1,
\]
where 
\begin{enumerate}
	\item $C^0(\Gamma,\C^\times) = \text{Hom}_\Z(C_0(\Gamma,\Z),\C^\times) \cong \{f: B \sqcup W \ra \C^\times\}$,
	\item $C^1(\Gamma,\C^\times)= \text{Hom}_\Z(C_1(\Gamma,\Z),\C^\times) \cong \{ \wt: E \ra \C^\times\}$,
	\item $(d f)(e) ={f(\b)}^{-1}{f (\bw)}$,
\end{enumerate}
where $\text{Hom}_\Z(\cdot,\C^\times)$ denotes the space of abelian group homomorphisms from $\cdot$ to $\C^\times$. Therefore, a $1$-cocycle is the same thing as an edge weight and two $1$-cocycles differ by a $1$-coboundary if they are gauge equivalent, so 
\[
\mathcal X_\Gamma = H^1(\Gamma,\C^\times).	
\]
With this identification, we can describe coordinates on $\mathcal X_\Gamma$. We describe this abstractly first and then explain what it means concretely. 

Since $H^1(\Gamma,\C^\times)=\mathrm{Hom}_\Z(H_1(\Gamma,\Z),\C^\times)$ is the algebraic torus with lattice of characters $H_1(\Gamma,\Z)$, its coordinate ring is 
\[
\C [H_1(\Gamma,\Z)] = \C\left[\sum_{\gamma \in H_1(\Gamma,\Z)} a_{\gamma}X_{\gamma}: a_{\gamma} = 0 \text{ for all but finitely many }\gamma\right]	
\] 
with multiplication given by $X_{\gamma_1} \cdot X_{\gamma_2}=X_{\gamma_1+\gamma_2}$, where $X_\gamma:\mathcal X_\Gamma \ra \C^\times$ is the character given by $X_\gamma([\wt])= [\wt](\gamma)$, i.e., by evaluating the cohomology class $[\wt]$ on the cycle $\gamma$. Explicitly, if $\gamma$ is the cycle ${\rm b}_1 \xrightarrow[]{e_1} {\rm w}_1 \xrightarrow[]{e_2} {\rm b}_2 \xrightarrow[]{e_3} {\rm w}_2 \xrightarrow[]{e_4} \cdots \xrightarrow[]{e_{2n-2}} {\rm b}_n \xrightarrow[]{e_{2n-1}} {\rm w}_n \xrightarrow[]{e_{2n}} {\rm b}_1 $ , then 
\[
[\wt](\gamma)=\frac{\wt({e_1}) \cdots \wt{(e_{2n-1}})}{ \wt{(e_2)} \cdots \wt{(e_{2n})}}.
\]
 In down to earth terms, this means that the $X_{\gamma}$'s are a set coordinates on $\mathcal X_\Gamma$ but they are not independent and have to satisfy some relations. If we choose a basis for $H_1(\Gamma,\Z)$, then we obtain a basis for the set of coordinates. There is no canonical choice for a basis but the following choice is convenient for computations. We identify faces of $\Gamma$ with the cycles given by their counterclockwise-oriented boundaries. Then each face $f$ determines a regular function $X_{f}$ on $\mathcal X_\Gamma$. Since $\sum_{f \in F} f=0$ in $H_1(\Gamma,\Z)$, we have $\prod_{f \in F}X_{f}=1$. Additionally if we choose two cycles $a,b$ in $\Gamma$ generating $H_1(\T,\Z)$, then $
	\{\text{all faces except one}, a , b\}$ is a basis for $H_1(\Gamma,\Z)$ and so \begin{equation} \label{h1:basis}
		\{X_f\}_{\text{all faces $f$ except one}} \sqcup \{X_a,X_b\}
	\end{equation}
		is a basis for the coordinate ring.

Let $\mathcal X_\Gamma(\Rpos)$ denote the positive real valued points of $\mathcal X_\Gamma$, i.e., the set of points $[\wt]$ where $X_\gamma([\wt]) \in \Rpos$ for all $\gamma \in H_1(\Gamma,\Z)$. 

\subsection{Zig-zag paths}
A \textit{zig-zag path} in $\Gamma$ is a cycle that turns maximally right at black vertices and maximally left at white vertices. We denote the set of zig-zag paths of $\Gamma$ by $Z_\Gamma$. The unique (modulo translation) convex integral polygon $N_\Gamma \subset  \R^2 $ whose counterclockwise-oriented boundary consists of $\{[\alpha] \in H_1(\T,\Z) \cong \Z^2  : \alpha \in  Z_\Gamma\}$ is called the \textit{Newton polygon} of $\Gamma$. For a side $S$ of $N$, let $Z_{\Gamma,S}$ denote the subset of zig-zag paths whose homology classes form the side $S$. 

Recall that $\pi:\R^2 \ra \T$ is the universal cover. We say that $\Gamma$ is \textit{minimal} if there are no zig-zag paths with zero homology, no lift of a zig-zag path has a self-intersection, and lifts of two zig-zag paths do not form \textit{parallel bigons}, i.e., they do no pass through two edges $e_1 \neq e_2$ of the biperiodic graph $\pi^{-1}(\Gamma) \subset \R^2$ with both lifts oriented from $e_1$ to $e_2$. Hereafter we assume that all our bipartite torus graphs are minimal.

\subsection{Moves}
\begin{figure}
	\begin{tikzpicture}[scale=0.6]
		\def\ep{0.2}

		\begin{scope}[shift={(-3,0)}]
			\def\r{2};
			% \draw[gray,dashed,-] (0,0) circle (\r cm);
			\fill[black!5] (0,0) circle (1.05*\r cm);
			\coordinate[nvert, label=right:${\rm v}$] (n1) at (0,0);
			\coordinate[] (n2) at (90:\r);
			\coordinate[] (n3) at (90+120:\r);
			\coordinate[] (n4) at (-30:\r);

			\draw[-]
			(n1) edge  (n2) edge (n3) edge (n4)
			;
		
			\draw[red,->,line width=\lw,rounded corners=\rc] (90+120-5:\r) -- (120:\ep) -- (95:\r);
			\draw[blue,->,line width=\lw,rounded corners=\rc] (90+120+5:\r) -- (270:\ep) -- (-35:\r);

			\node[](no) at (120:1.5) {$\delta$};
			\node[](no) at (-60:1.5) {$\gamma$};
		  \end{scope}

	  \end{tikzpicture}
	  \caption{The local pairing at a trivalent vertex $\rm v$ (which may be black or white) is defined by $\epsilon_{\rm v}(\gamma,\delta)=\frac 1 2$ and bilinearity and antisymmetry.} \label{fig:eps}
\end{figure}
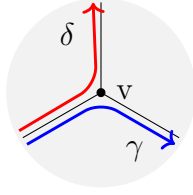

\begin{figure}
	\begin{tikzpicture}[scale=0.6]
		\draw[] (2.5-5,-2) -- (2.5-5,7);

		\node[](no) at (-5.5,-4) {(a)};
		\node[](no) at (-5.5+5.5,-4) {(b)};
		\node[](no) at (-5.5+10.5,-4) {(c)};
		\node[](no) at (-5.5+15.5,-4) {(d)};
		\node[](no) at (-5.5+20.5,-4) {(e)};

		\node[](no) at (-5.5+5.5,3.5) {$f_1$};
		\node[](no) at (-5.5+10.5,3.5) {$f_2$};
		\node[](no) at (-5.5+15.5,3.5) {$f_3$};
		\node[](no) at (-5.5+20.5,3.5) {$f_4$};

		\node[](no) at (-5.5+5.5,-2.5) {$f_1'+f'$};
		\node[](no) at (-5.5+10.5,-2.5) {$f_2'$};
		\node[](no) at (-5.5+15.5,-2.5) {$f_3'+f'$};
		\node[](no) at (-5.5+20.5,-2.5) {$f_4'$};

		\begin{scope}[shift={(-5.5,0)},rotate=45]
			\def\r{2};
			\fill[black!5] (0,0) circle (\r cm);
			\coordinate[wvert,label=right: $\bw_2$] (n1) at (0:\r);
			\coordinate[wvert,label=left: $\bw_1$] (n2) at (0+90:\r);
			\coordinate[wvert,label=left: $\bw_4$] (n3) at (0+180:\r);
			\coordinate[wvert,label=right: $\bw_3$] (n4) at (0+270:\r);

			\coordinate[bvert,label=above: $\b_1'$] (b1) at (0:0.5*\r);
			\coordinate[bvert,label=below: $\b_2'$] (b2) at (0+180:0.5*\r);
			
			\draw[-]
							(n1) --  (b1) --  (n2)--(b2)--(n4) -- (b1)
							(b2)--(n3)
							;
							\node[red] (no) at (0,0) {$f'$};	
							\node[red] (no) at (45:1.7) {$f_2'$};
							\node[red] (no) at (45+90:1.7) {$f_1'$};
							\node[red] (no) at (45+180:1.7) {$f_4'$};
							\node[red] (no) at (45-90:1.7) {$f_3'$};			
		
			\end{scope}

			\begin{scope}[shift={(-5.5,6)}
				,rotate=45+90
				]
				\def\r{2};
				\fill[black!5] (0,0) circle (\r cm);
				\coordinate[wvert,label=left: $\bw_1$] (n1) at (0:\r);
				\coordinate[wvert,label=left: $\bw_4$] (n2) at (0+90:\r);
				\coordinate[wvert,label=right: $\bw_3$] (n3) at (0+180:\r);
				\coordinate[wvert,label=right: $\bw_2$] (n4) at (0+270:\r);

				\coordinate[bvert,label=above: $\b_1$] (b1) at (0:0.5*\r);
				\coordinate[bvert,label=below: $\b_2$] (b2) at (0+180:0.5*\r);
				
				\draw[-]
								(n1) --  (b1) --  (n2)--(b2)--(n4) -- (b1)
								(b2)--(n3)
								;
		
				\node[red] (no) at (0,0) {$f$};	
				\node[red] (no) at (45:1.7) {$f_1$};
				\node[red] (no) at (45+90:1.7) {$f_4$};
				\node[red] (no) at (45+180:1.7) {$f_3$};
				\node[red] (no) at (45-90:1.7) {$f_2$};			
							
				\end{scope}

		% \node[](no) at (0,-3){(M1) The spider move.};
		\begin{scope}[shift={(0,0)}
			,rotate=45
			]
			\def\ep{0.2}
		  \def\r{2};
		  % \draw[gray,dashed,-] (0,0) circle (\r cm);
		  \fill[black!5] (0,0) circle (1.05*\r cm);
		  \coordinate[wvert] (n1) at (0:\r);
		  \coordinate[wvert] (n2) at (0+90:\r);
		  \coordinate[wvert] (n3) at (0+180:\r);
		  \coordinate[wvert] (n4) at (0+270:\r);
		  
		  \coordinate[bvert] (b1) at (0:0.5*\r);
		  \coordinate[bvert] (b2) at (0+180:0.5*\r);
		  
		  \draw[-]
		  (n1) --  (b1) --  (n2)--  (b2)-- (n4) --  (b1)
		  (b2)--(n3)
		  ;
		  \coordinate[] (t1) at (15:\r);
		  \coordinate[] (t2) at (120-45:\r);
		  \coordinate[] (t3) at (150-45:\r);
		  \coordinate[] (t4) at (210-45:\r);
		  \coordinate[] (t5) at (240-45:\r);
		  \coordinate[] (t6) at (300-45:\r);
		  \coordinate[] (t7) at (330-45:\r);
		  \coordinate[] (t8) at (30-45:\r);
		  
		  \draw[red,->,line width=\lw,rounded corners=\rc] (-2,0.1+\ep) -- (-1+\ep,\ep) -- (0,-2+\ep)--(1-\ep,-\ep)-- (-0.1-\ep,2);
		 
		\end{scope}
		
		\begin{scope}[shift={(0,6)}
			,rotate=180+90+45
			]
			\def\ep{0.2}
		  \def\r{2};
		  % \draw[gray,dashed,-] (0,0) circle (\r cm);
		  \fill[black!5] (0,0) circle (1.05*\r cm);
		  \coordinate[wvert] (n1) at (0:\r);
		  \coordinate[wvert] (n2) at (0+90:\r);
		  \coordinate[wvert] (n3) at (0+180:\r);
		  \coordinate[wvert] (n4) at (0+270:\r);
		  
		  \coordinate[bvert] (b1) at (0:0.5*\r);
		  \coordinate[bvert] (b2) at (0+180:0.5*\r);
		  
		  \draw[-]
		  (n1) --  (b1) --  (n2)--  (b2)-- (n4) --  (b1)
		  (b2)--(n3)
		  ;
		  \coordinate[] (t1) at (15:\r);
		  \coordinate[] (t2) at (120-45:\r);
		  \coordinate[] (t3) at (150-45:\r);
		  \coordinate[] (t4) at (210-45:\r);
		  \coordinate[] (t5) at (240-45:\r);
		  \coordinate[] (t6) at (300-45:\r);
		  \coordinate[] (t7) at (330-45:\r);
		  \coordinate[] (t8) at (30-45:\r);
		  
		  \draw[red,<-,line width=\lw,rounded corners=\rc] (-2,-0.1-\ep) -- (-1-\ep,-\ep) -- (-\ep-0.1,-2);
		 
		\end{scope}

		\draw[] (2.5,-2) -- (2.5,7);

		%% f2

		\begin{scope}[shift={(5,0)}
			,rotate=45
			]
			\def\ep{0.2}
		  \def\r{2};
		  % \draw[gray,dashed,-] (0,0) circle (\r cm);
		  \fill[black!5] (0,0) circle (1.05*\r cm);
		  \coordinate[wvert] (n1) at (0:\r);
		  \coordinate[wvert] (n2) at (0+90:\r);
		  \coordinate[wvert] (n3) at (0+180:\r);
		  \coordinate[wvert] (n4) at (0+270:\r);
		  
		  \coordinate[bvert] (b1) at (0:0.5*\r);
		  \coordinate[bvert] (b2) at (0+180:0.5*\r);
		  
		  \draw[-]
		  (n1) --  (b1) --  (n2)--  (b2)-- (n4) --  (b1)
		  (b2)--(n3)
		  ;
		  \coordinate[] (t1) at (15:\r);
		  \coordinate[] (t2) at (120-45:\r);
		  \coordinate[] (t3) at (150-45:\r);
		  \coordinate[] (t4) at (210-45:\r);
		  \coordinate[] (t5) at (240-45:\r);
		  \coordinate[] (t6) at (300-45:\r);
		  \coordinate[] (t7) at (330-45:\r);
		  \coordinate[] (t8) at (30-45:\r);
		  
		  \draw[red,->,line width=\lw,rounded corners=\rc]  (\ep+0.1,2) -- (1+\ep,\ep) --(2,+\ep+0.1);
		 
		\end{scope}
		
		\begin{scope}[shift={(5,6)}
			,rotate=180+90+45
			]
			\def\ep{0.2}
		  \def\r{2};
		  % \draw[gray,dashed,-] (0,0) circle (\r cm);
		  \fill[black!5] (0,0) circle (1.05*\r cm);
		  \coordinate[wvert] (n1) at (0:\r);
		  \coordinate[wvert] (n2) at (0+90:\r);
		  \coordinate[wvert] (n3) at (0+180:\r);
		  \coordinate[wvert] (n4) at (0+270:\r);
		  
		  \coordinate[bvert] (b1) at (0:0.5*\r);
		  \coordinate[bvert] (b2) at (0+180:0.5*\r);
		  
		  \draw[-]
		  (n1) --  (b1) --  (n2)--  (b2)-- (n4) --  (b1)
		  (b2)--(n3)
		  ;
		  \coordinate[] (t1) at (15:\r);
		  \coordinate[] (t2) at (120-45:\r);
		  \coordinate[] (t3) at (150-45:\r);
		  \coordinate[] (t4) at (210-45:\r);
		  \coordinate[] (t5) at (240-45:\r);
		  \coordinate[] (t6) at (300-45:\r);
		  \coordinate[] (t7) at (330-45:\r);
		  \coordinate[] (t8) at (30-45:\r);
		  
		  \draw[red,->,line width=\lw,rounded corners=\rc] (-2,0.1+\ep) -- (-1-\ep,\ep) -- (-\ep-0.1,2);
		 
		\end{scope}

		\draw[] (2.5+5,-2) -- (2.5+5,7);

		\begin{scope}[shift={(10,0)}
			,rotate=180+45
			]
			\def\ep{0.2}
		  \def\r{2};
		  % \draw[gray,dashed,-] (0,0) circle (\r cm);
		  \fill[black!5] (0,0) circle (1.05*\r cm);
		  \coordinate[wvert] (n1) at (0:\r);
		  \coordinate[wvert] (n2) at (0+90:\r);
		  \coordinate[wvert] (n3) at (0+180:\r);
		  \coordinate[wvert] (n4) at (0+270:\r);
		  
		  \coordinate[bvert] (b1) at (0:0.5*\r);
		  \coordinate[bvert] (b2) at (0+180:0.5*\r);
		  
		  \draw[-]
		  (n1) --  (b1) --  (n2)--  (b2)-- (n4) --  (b1)
		  (b2)--(n3)
		  ;
		  \coordinate[] (t1) at (15:\r);
		  \coordinate[] (t2) at (120-45:\r);
		  \coordinate[] (t3) at (150-45:\r);
		  \coordinate[] (t4) at (210-45:\r);
		  \coordinate[] (t5) at (240-45:\r);
		  \coordinate[] (t6) at (300-45:\r);
		  \coordinate[] (t7) at (330-45:\r);
		  \coordinate[] (t8) at (30-45:\r);
		  
		  \draw[red,->,line width=\lw,rounded corners=\rc] (-2,0.1+\ep) -- (-1+\ep,\ep) -- (0,-2+\ep)--(1-\ep,-\ep)-- (-0.1-\ep,2);
		 
		\end{scope}
		
		\begin{scope}[shift={(10,6)}
			,rotate=90+45
			]
			\def\ep{0.2}
		  \def\r{2};
		  % \draw[gray,dashed,-] (0,0) circle (\r cm);
		  \fill[black!5] (0,0) circle (1.05*\r cm);
		  \coordinate[wvert] (n1) at (0:\r);
		  \coordinate[wvert] (n2) at (0+90:\r);
		  \coordinate[wvert] (n3) at (0+180:\r);
		  \coordinate[wvert] (n4) at (0+270:\r);
		  
		  \coordinate[bvert] (b1) at (0:0.5*\r);
		  \coordinate[bvert] (b2) at (0+180:0.5*\r);
		  
		  \draw[-]
		  (n1) --  (b1) --  (n2)--  (b2)-- (n4) --  (b1)
		  (b2)--(n3)
		  ;
		  \coordinate[] (t1) at (15:\r);
		  \coordinate[] (t2) at (120-45:\r);
		  \coordinate[] (t3) at (150-45:\r);
		  \coordinate[] (t4) at (210-45:\r);
		  \coordinate[] (t5) at (240-45:\r);
		  \coordinate[] (t6) at (300-45:\r);
		  \coordinate[] (t7) at (330-45:\r);
		  \coordinate[] (t8) at (30-45:\r);
		  
		  \draw[red,<-,line width=\lw,rounded corners=\rc] (-2,-0.1-\ep) -- (-1-\ep,-\ep) -- (-\ep-0.1,-2);
		 
		\end{scope}

		\draw[] (12.5,-2) -- (12.5,7);
		\begin{scope}[shift={(15,0)}
			,rotate=180+45
			]
			\def\ep{0.2}
		  \def\r{2};
		  % \draw[gray,dashed,-] (0,0) circle (\r cm);
		  \fill[black!5] (0,0) circle (1.05*\r cm);
		  \coordinate[wvert] (n1) at (0:\r);
		  \coordinate[wvert] (n2) at (0+90:\r);
		  \coordinate[wvert] (n3) at (0+180:\r);
		  \coordinate[wvert] (n4) at (0+270:\r);
		  
		  \coordinate[bvert] (b1) at (0:0.5*\r);
		  \coordinate[bvert] (b2) at (0+180:0.5*\r);
		  
		  \draw[-]
		  (n1) --  (b1) --  (n2)--  (b2)-- (n4) --  (b1)
		  (b2)--(n3)
		  ;
		  \coordinate[] (t1) at (15:\r);
		  \coordinate[] (t2) at (120-45:\r);
		  \coordinate[] (t3) at (150-45:\r);
		  \coordinate[] (t4) at (210-45:\r);
		  \coordinate[] (t5) at (240-45:\r);
		  \coordinate[] (t6) at (300-45:\r);
		  \coordinate[] (t7) at (330-45:\r);
		  \coordinate[] (t8) at (30-45:\r);
		  
		  \draw[red,->,line width=\lw,rounded corners=\rc]  (\ep+0.1,2) -- (1+\ep,\ep) --(2,+\ep+0.1);
		 
		\end{scope}
		
		\begin{scope}[shift={(15,6)}
			,rotate=90+45
			]
			\def\ep{0.2}
		  \def\r{2};
		  % \draw[gray,dashed,-] (0,0) circle (\r cm);
		  \fill[black!5] (0,0) circle (1.05*\r cm);
		  \coordinate[wvert] (n1) at (0:\r);
		  \coordinate[wvert] (n2) at (0+90:\r);
		  \coordinate[wvert] (n3) at (0+180:\r);
		  \coordinate[wvert] (n4) at (0+270:\r);
		  
		  \coordinate[bvert] (b1) at (0:0.5*\r);
		  \coordinate[bvert] (b2) at (0+180:0.5*\r);
		  
		  \draw[-]
		  (n1) --  (b1) --  (n2)--  (b2)-- (n4) --  (b1)
		  (b2)--(n3)
		  ;
		  \coordinate[] (t1) at (15:\r);
		  \coordinate[] (t2) at (120-45:\r);
		  \coordinate[] (t3) at (150-45:\r);
		  \coordinate[] (t4) at (210-45:\r);
		  \coordinate[] (t5) at (240-45:\r);
		  \coordinate[] (t6) at (300-45:\r);
		  \coordinate[] (t7) at (330-45:\r);
		  \coordinate[] (t8) at (30-45:\r);
		  
		  \draw[red,->,line width=\lw,rounded corners=\rc] (-2,0.1+\ep) -- (-1-\ep,\ep) -- (-\ep-0.1,2);
		 
		\end{scope}

	  \end{tikzpicture}
	  \caption{(a) Labeling of the graph near a square face $f$ and (b)--(e) the correspondence $s_*$ between cycles in $\Gamma$ and $\Gamma'$.} \label{fig:cycle_corresponds}\label{fig:moves}
\end{figure}
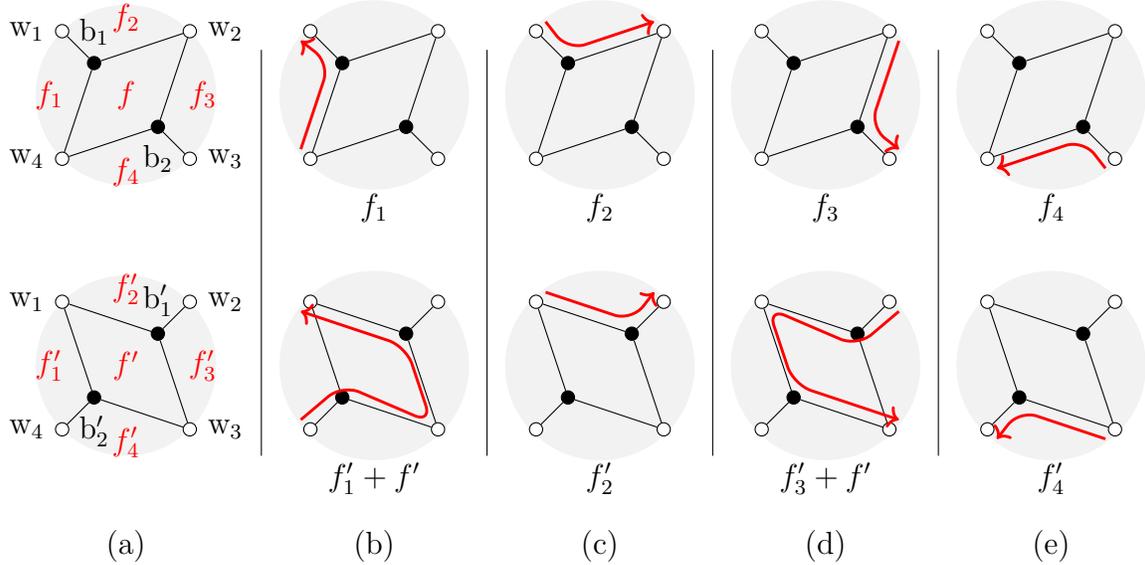

There are two local moves on bipartite torus graphs shown in Figure~\ref{fig:squareintro}. Each move $\Gamma \stackrel{s}{\rightsquigarrow} \Gamma'$ induces a pair of maps:
\begin{enumerate}
    \item An isomorphism $s_*:H_1(\Gamma,\Z) \ra H_1(\Gamma',\Z)$ such that $s_*$ restricts to a bijection between $Z_{\Gamma}$ and $Z_{\Gamma'}$ and such that $[s_*(\gamma)] = [\gamma] \in H_1(\T,\Z)$, 
\item A bijection $\mu_s: \mathcal X_{\Gamma}(\Rpos) \ra \mathcal X_{\Gamma'}(\Rpos)$.
\end{enumerate}

We discuss the case of the square move at a face $f$ in more detail (see~\cite{GK, GLSS}). Let $\Gamma_f$ and $\Gamma_f'$ denote the portions of the graphs shown in Figure~\ref{fig:moves}. Suppose the vertices and faces near $f$ are labeled as in Figure~\ref{fig:moves} (left). Let $\mathbf M$ denote the set $\{\w_1,\w_2,\w_3,\w_4\}$. Then each cycle in $\Gamma$ (resp., $\Gamma'$) restricts to a relative cycle in $H_1(\Gamma_f,\mathbf M,\Z)$ (resp., $H_1(\Gamma_f',\mathbf M,\Z)$) and $f_1,f_2,f_3,f_4$ (resp., $f_1',f_2',f_3',f_4'$) is a basis for $H_1(\Gamma_f,\mathbf M,\Z)$ (resp., $H_1(\Gamma_f',\mathbf M,\Z)$). Note that $f = -f_1-f_2-f_3-f_4$ and similarly for $f'$.

In these bases, $s_*$ is given by the formula 
\begin{equation}\label{eq:fmut}
	s_*(f_i) =f_i'  + \text{max}(0, \langle f_i,f \rangle ) f',
\end{equation}
where $\langle \cdot,\cdot \rangle:H_1(\Gamma,\Z) \times H_1(\Gamma,\Z)$ is an antisymmetric bilinear form on cycles defined by the intersection pairing on the conjugate surface (see~\cite[Section 1.1.1]{GK}). It can be computed by 
\[
\langle \gamma,\delta \rangle =  \sum_{{\rm b} \in B } \epsilon_{\b}(\gamma,\delta)	-\sum_{{\rm w} \in W } \epsilon_{\bw}(\gamma,\delta),	
\]
where $\epsilon_{\rm v}(\gamma,\delta)$ is the local pairing shown in Figure~\ref{fig:eps}. Since all the vertices in $G^\square$ are trivalent, we do not give the general rule (see~\cite[Appendix]{GK}).

For example, we have 
\[
\langle f_1,f\rangle = \epsilon_{\b_1}(f_1,f)-\epsilon_{\bw_4}(f_1,f) = \frac 1 2 -\left( -\frac 1 2 \right) =1,	
\]
so $s_*(f_1) = f_1' + f'$. The cycles $s_*(f_i), i \in \{1,2,3,4\}$, are shown in Figure~\ref{fig:moves}(b)--(e). Note that
\[
	s_*(f)=s_*(-f_1-f_2-f_3-f_4)=-(f_1'+f')-f_2'-(f_3'+f')-f_4' =-f'.
\]

The bijection $\mu_s: \mathcal X_\Gamma (\Rpos) \ra \mathcal X_{\Gamma'}(\Rpos)$ is given by
\begin{equation} \label{eq:mu}
	X_{s_*(\gamma)}(\mu_s([\wt])) = X_\gamma([\wt]) (1+X_f([\wt]))^{-\langle \gamma,f \rangle},
\end{equation}
or more succinctly, $\mu^*_sX_{s_*(\gamma)} =X_\gamma (1+X_f)^{-\langle \gamma,f \rangle}$, which is the formula for mutation of $X$ cluster variables. For example, 
\begin{align*}
	\mu^*_s X_{f'} =\mu^*_s X_{s_*(-f)}= X_{-f} (1+X_f)^{-\langle -f,f \rangle}=X_f^{-1},
\end{align*}
and \[
	\mu^*_s X_{f_1'} = \mu^*_s X_{s_*(f_1+f)}=X_{f_1+f}(1+X_f)^{-1} = X_{f_1} (1+X_f^{-1})^{-1}.
	\]

\subsection{The mapping from Ising models to dimer models}

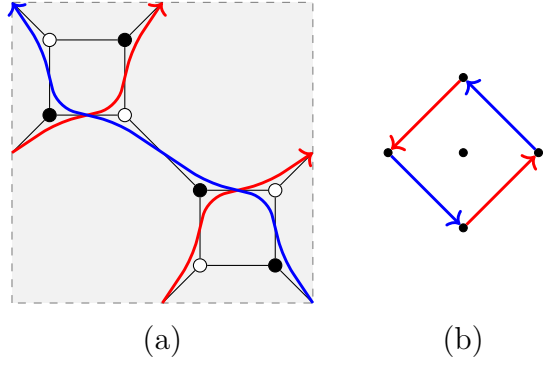
\begin{figure}
	\begin{tikzpicture}[scale=0.5]

		 \begin{scope}[shift={(10,0)},rotate=0]
\def\ep{0.2}
			% \draw[gray,dashed,-] (0,0) circle (\r cm);
			\fill[black!5] (0,0) rectangle (8,8);
			\draw[gray,dashed] (0,0) rectangle (8,8);

			\coordinate[bvert] (b1) at (1,5);
			\coordinate[bvert] (b2) at (1+2,5+2);
			\coordinate[bvert] (b3) at (5,3);
			\coordinate[bvert] (b4) at (7,1);
  
			\coordinate[wvert] (w1) at (1,7);
			\coordinate[wvert] (w2) at (1+2,5);
			\coordinate[wvert] (w3) at (5,1);
			\coordinate[wvert] (w4) at (7,3);
  
			\draw[] (b1) -- (0,4) 
			(b1) -- (w2)
			(b1) --(w1);
  
			\draw[] (b2) -- (w1)
			(b2) -- (4,8)
			(b2) -- (w2);
  
			\draw[] (b3) -- (w3)
			(b3) -- (w2)
			(b3) -- (w4) ;
  
			\draw[] (b4) -- (w3)
			(b4) -- (8,0)
			(b4) -- (w4) ;
  
			\draw[] (w4) -- (8,4)
			(0,8) -- (w1)
			(4,0) -- (w3);

			\draw[red,->,line width=\lw,rounded corners=\rc] (0,4) -- (1+\ep,5-\ep) -- (3-\ep,5+\ep) -- (3+\ep,7-\ep) -- (4,8); 

			\draw[red,->,line width=\lw,rounded corners=\rc] (4,0) -- (5-\ep,1+\ep) -- (5+\ep,3-\ep) -- (7-\ep,3+\ep) -- (8,4); 

			\draw[blue,->,line width=\lw,rounded corners=\rc] (8,0) -- (7+\ep,1+\ep) -- (7-\ep,3-\ep) -- (5+\ep,3+\ep) -- (3-\ep,5-\ep) -- (1+\ep,5+\ep) -- (1-\ep,7-\ep) -- (0,8);
			\node[](no) at (4,-1) {(a)};
		   \end{scope}

		 \begin{scope}[shift={(22,4)},rotate=0,scale=2]
			\coordinate[nvert] (00) at (0,0);
			\coordinate[nvert] (10) at (1,0);
			\coordinate[nvert] (01) at (0,1);
			\coordinate[nvert] (m10) at (-1,0);
			\coordinate[nvert] (0m1) at (0,-1);
\draw[red,->, line width = \lw] (0m1)--(10);
\draw[red,->, line width = \lw] (01)--(m10);
\draw[blue,->, line width = \lw] (10)--(01);
\draw[blue,->, line width = \lw] (m10)--(0m1);
\node[](no) at (0,-2.5) {(b)};
		  \end{scope}

	  \end{tikzpicture}
	  \caption{(a) Two of the zig-zag paths of the graph $G^\square$ from Figure~\ref{fig:isingintro}(b) and (b) its Newton polygon.} \label{fig:ising_torus_bijection}
\end{figure}

Recall the mapping from Ising models to dimer models from Figure~\ref{fig:isingintro}(b). There is a natural bijection between $Z_G$ and $Z_{G^\square}$ that preserves homology classes (compare Figure~\ref{fig:ising_torus}(a) and Figure~\ref{fig:ising_torus_bijection}(a)). Therefore, $N_{G^\square}=N_G$. 

A Y-$\Delta$ move $G_1 \ra G_2$ can be realized as a sequence of moves $G_1^\square \ra G_2^\square$ \cite[Figure 6]{KP} (see also~\cite[Section 4.3]{AGPR}), i.e., we have a commuting diagram
	\[
		\begin{tikzcd}
			\mathcal I_{G_1}(\Rpos) \arrow[r,"\text{Y-$\Delta$ move}"]\arrow[d,hook] &\mathcal I_{G_2}(\Rpos) \arrow[d,hook]\\
			\mathcal X_{G_1^\square}(\Rpos) \arrow[r,"\text{moves}"] & \mathcal X_{G_2^\square}(\Rpos)
		\end{tikzcd},
		\]	
 and the dual Ising model gives the dimer model $(\overline{G^\square},[\overline{\wt^\square}])$. Therefore, the mapping is compatible with respect to moves on both models.

%The main goal of the paper is to characterize the subset $\mathcal I_G(\Rpos) \subset \mathcal X_{G^\square}(\Rpos)$.

\subsection{The spectral transform}

Let $(\Gamma,[\wt])$ be a dimer model in $\T$. We assign to each edge $e =\{\b,\bw\}$ of $\Gamma$ a monomial $\phi(e) := z^iw^j$ where $i$ (resp.,  $j$) records the number of signed intersections of $e$ with $\gamma_w$ (resp., $\gamma_z$) when we orient $e$ from $\b$ to $\bw$.  

A \textit{Kasteleyn sign} on $\Gamma$ is a $1$-cohomology class $[\kappa] \in H^1(\Gamma, \{\pm 1\})$ such that for any face $f \in F$, 
\[
 [\kappa](f) = (-1)^{\frac{\# \partial f}{2}+1},	
\]
where $e \in f$ means $e$ is incident to $f$. Here $\{\pm 1\}$ is the multiplicative group of square roots of unity. There are four choices of $[\kappa]$ corresponding to the four elements of $H^1(\T,\{\pm 1\})$. 
If we use the basis (\ref{h1:basis}), a choice of $[\kappa]$ is equivalent to a choice of $(X_a,X_b) \in \{\pm 1\}^2$. 

Let $\wt$ be a $1$-cocycle representing the cohomology class $[\wt]$ and let $\kappa:E \ra \{\pm 1\}$ be a $1$-cocycle representing $[\kappa]$. The \textit{Kasteleyn matrix} $K_{(\Gamma, \wt, \kappa)}(z,w)$ is the Laurent-polynomial-valued matrix with rows and columns indexed by white and black vertices respectively defined by
\[
K_{(\Gamma, \wt, \kappa)}(z,w)_{\bw \b} := \sum_{e=\{\b,\bw\} }\wt(e) \kappa(e) \phi(e),	
\] 
where the sum is over all edges with endpoints $\b$ and $\bw$. The determinant \[P_{(\Gamma, \wt, \kappa)}(z,w) := \det K_{(\Gamma, \wt, \kappa)}(z,w)\] is a Laurent polynomial called the \textit{characteristic polynomial}. The \textit{Newton polygon} of $P_{(\Gamma, \wt, \kappa)}(z,w)$ is defined as
\[
\text{convex-hull} \{ (i,j) \in \R^2 : \text{coefficient of $z^iw^j$ is nonzero in $P_{(\Gamma, \wt, \kappa)}(z,w)$}\}.	
\]
For minimal $\Gamma$, the Newton polygon of $P_{(\Gamma, \wt, \kappa)}(z,w)$ coincides with the Newton polygon of $\Gamma$, whence the name. The vanishing locus $C^\circ:=\{(z,w)\in (\C^\times)^2 : P_{(\Gamma, \wt, \kappa)}(z,w)=0\}$ is called the \textit{open spectral curve}. 

The \textit{amoeba} $\mathbb A(C^\circ)$ is defined to be the image of $C^\circ$ under the map $\operatorname{Log}:(\C^\times)^2 \ra \R^2$ defined by $(z,w) \mapsto (\log |z| , \log |w|)$. A curve $C^\circ$ defined by a real Laurent polynomial $P(z,w)$ is said to be a \textit{Harnack curve} if the restriction
\[
\operatorname*{Log}: C^\circ \ra \mathbb A(C^\circ)	
\]
is $2:1$ over the interior of $\mathbb A(C^\circ)$. For Harnack curves, the boundary of $\mathbb A(C^\circ)$ is the image of the real points $C^\circ(\R)$.

\begin{theorem}[\cite{KSh,KOS}]
	The open spectral curve $C^\circ$ of a dimer model is a Harnack curve.
\end{theorem}

The Newton polygon $N$ defines a compactification 
\[
\mathcal N:= \overline{ \{[z^i w^j]_{(i,j) \in N \cap \Z^2}: (z,w) \in (\C^\times)^2 \}} \subset \C\mathbb P^{\# N \cap \Z^2 -1}
 \]
of $(\C^\times)^2$ called a \textit{toric surface}. The boundary 
\[
\mathcal N - (\C^\times)^2 = \bigcup_{S \text{ side of }N}D_S	
\]
where each $D_S \cong \C \mathbb P^1$ is called a \textit{line at infinity}. Taking closure, we get a compactification 
\[
C:= \overline {C^\circ} \subset \mathcal N 
\]
called the \textit{spectral curve}. The points $C-C^\circ$ are called \textit{points at infinity}; we call the points in $C \cap D_S$ the \textit{points at infinity corresponding to the side $S$}. The number of such points is equal to the number of primitive vectors in $S$. 

Points at infinity corresponding to the side $S$ are in bijection with \textit{tentacles} of the amoeba in the direction of the outward-pointing normal to $S$. Each tentacle is asymptotic to a line $i x + jy+c=0$ where $(i,j)$ is the primitive vector along the side $S$ of $N$. There is a bijection \[\nu_{\Gamma} : Z_{\Gamma} \ra \bigsqcup_{S \text{ side of }N} C \cap D_S\] which can be described as follows: The tentacle corresponding to a zig-zag path $\alpha$ with $[\alpha]$ along $S$ is the one asymptotic to the line 
\begin{equation}\label{eq:bij}
	ix+jy+\log|X_\alpha([\wt])|=0.
\end{equation}
Recall that $Z_{\Gamma,S}$ is the subset of zig-zag paths corresponding to $S$. Note that by definition, this bijection has the property that $\nu_\Gamma(Z_{\Gamma,S}) = C \cap D_S$.

The real locus of a Harnack curve has $g+1$ components called \textit{ovals}, one of which corresponds to the boundary of $N$ and is called the \textit{outer} oval, and $g$ \textit{interior} ovals that are in bijection with the interior lattice points in $N$.  

By definition of the spectral curve, the Kasteleyn matrix $K_{(\Gamma, \wt, \kappa)}(z,w)$ has nonzero kernel and cokernel over $C^\circ$. If $(z,w)$ is a smooth point of $C^\circ$ then the kernel and cokernel of $K_{(\Gamma, \wt, \kappa)}(z,w)$ are $1$-dimensional. Hence, if $C^\circ$ is smooth then the kernel and cokernel are line bundles. 

For a white vertex $\bw$, the image of the function $\delta_{\bw} \in \C^{W}$ in $\operatorname{coker} K_{(\Gamma, \wt, \kappa)}(z,w)$ defines a section of the cokernel. This section vanishes on a set of $g$ points $\{(p_i,q_i)\}_{i=1}^g$ in $\C^\circ$, where $g$ is the arithmetic genus of $C$ and is equal to the number of interior points in $N \cap \Z^2$. We call 
\begin{equation} \label{eq:div}
D_\bw := \sum_{i=1}^g (p_i,q_i)	
\end{equation}
the \textit{divisor of $\bw$}. By the same construction applied to $K_{(\Gamma, \wt, \kappa)}(z,w)^T$, we define degree-$g$ divisors $D_\b$ for black vertices $\b \in B$. These divisors can be described more explicitly as follows. Let $Q_{(\Gamma, \wt, \kappa)}(z,w)$ denote the cofactor matrix of $K_{(\Gamma, \wt, \kappa)}(z,w)$. The divisor $D_\b$ (resp., $D_\bw$) is given by the vanishing of the $\b$-row (resp., $\bw$-column) of $Q_{(\Gamma, \wt, \kappa)}(z,w)$.

 A degree-$g$ divisor $D$ is called a \textit{standard} divisor if $D$ contains one point in each interior oval. If $C$ is smooth, the set of standard divisors is a component of the real locus of the Jacobian of $C$ and is a $g$-dimensional real torus.

\begin{theorem}[\cite{KO}]
	Let $(\Gamma,[\wt])$ be a dimer model in $\T$ and let $\mathrm{v}$ be a vertex of $\Gamma$. The spectral curve $C$ is a Harnack curve and the divisor $D_\mathrm{v}$ is a standard divisor. 
\end{theorem}

\begin{remark}\label{rem:singular}
	The only possible singularities of Harnack curves occur when some ovals shrink to zero size forming isolated real nodes or when tentacles in the same direction merge \cite{Mik}. Suppose $C$ is singular with $b$ nodes. Let $\tilde C$ denote its desingularization which has genus $\tilde g = g-b$. If $\pi:\tilde C \ra C$ is the canonical map gluing pairs of points into nodes, then $\pi^{-1}(D_{\rm v})$ has $g+b$ points: each smooth point lifts to one point and each node to two points. Let $\tilde D_{\rm v}$ denote the $\tilde g$ points of $\pi^{-1}(D_{\rm v})$ corresponding to the smooth points. All the constructions work for singular $C$ as for smooth $C$ if we replace $C$ with $\tilde C$ and $D_{\rm v}$ with $\tilde D_{\rm v}$, and we do this without further mention below.
\end{remark}

\begin{figure}
	\begin{tikzpicture}[scale=0.5]
		\def\ep{0.3}
		\def\xsh{1}
	
		\begin{scope}[shift={(0,0)},rotate=0]

			% \draw[gray,dashed,-] (0,0) circle (\r cm);
			\fill[black!5] (0,0) rectangle (8,8);
			\draw[gray,dashed] (0,0) rectangle (8,8);

			\coordinate[bvert] (b1) at (1,5);
			\coordinate[bvert] (b2) at (1+2,5+2);
			\coordinate[bvert,label=above:$\b$] (b3) at (5,3);
			\coordinate[bvert] (b4) at (7,1);
  
			\coordinate[wvert] (w1) at (1,7);
			\coordinate[wvert,label=right:$\bw$] (w2) at (1+2,5);
			\coordinate[wvert] (w3) at (5,1);
			\coordinate[wvert] (w4) at (7,3);
  
			\draw[] (b1) -- (0,4) 
			(b1) --node[below]{$s_2$} (w2)
			(b1) --node[left]{$c_2$} (w1);
  
			\draw[] (b2) --node[above]{$s_2$} (w1)
			(b2) -- (4,8)
			(b2) --node[right]{$-c_2$} (w2);
  
			\draw[] (b3) --node[left]{$-s_1$} (w3)
			(b3) -- (w2)
			(b3) --node[above]{$c_1$} (w4) ;
  
			\draw[] (b4) --node[below]{$c_1$} (w3)
			(b4) -- (8,0)
			(b4) --node[right]{$s_1$} (w4) ;
  
			\draw[] (w4) -- (8,4)
			(0,8) -- (w1)
			(4,0) -- (w3);
			
			\node[] (no) at (4,8.5) {$w$};
			\node[] (no) at (8.5,4) {$\frac 1 z$};
			\node[] (no) at (8.5,0) {$\frac z w$};
			\node[](no) at (4,-1) {(a)};
		   \end{scope}
\node (no) at (18,4) {\includegraphics[scale=0.37]{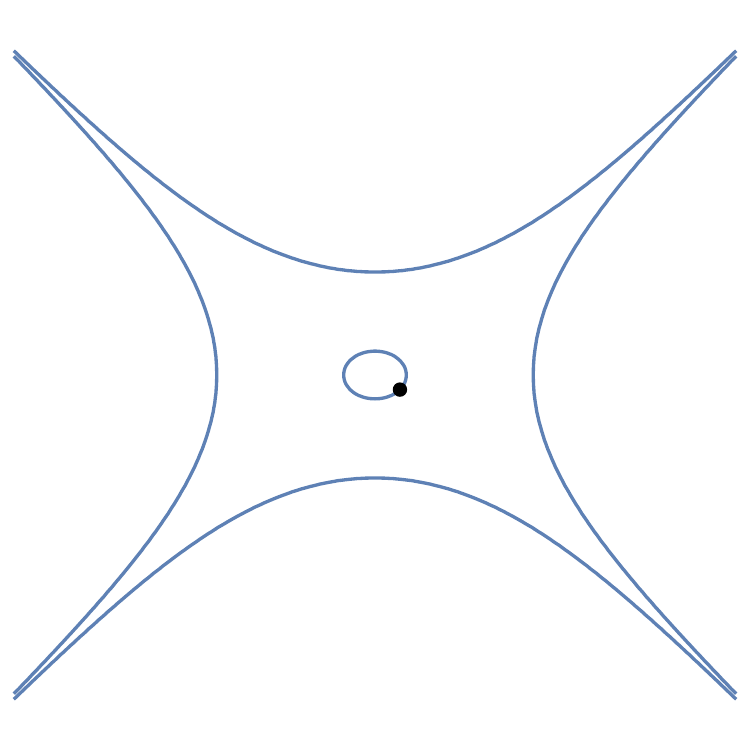}};
			\node[](no) at (18,-1) {(b)};
	
	  \end{tikzpicture}
	  \caption{(a) A Kasteleyn sign and $\phi$ for the dimer model in Figure~\ref{fig:isingintro}(b) and (b) the amoeba of $C$ and the standard divisor $D_{\bw}$ for $c_1=\frac{1}{\sqrt 2}$ and $c_2 = \frac{\sqrt 3}{2}$.} \label{fig:kasmat}
\end{figure}
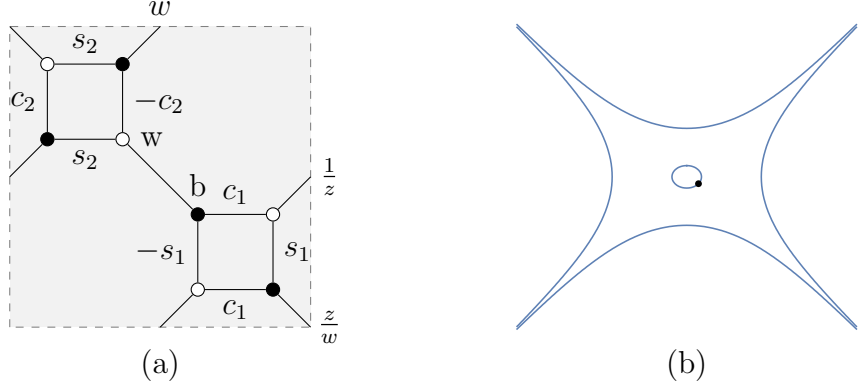

\begin{example}
	Let $(G^\square,[\wt^\square])$ be the dimer model from Figure~\ref{fig:isingintro}(c). Consider the Kasteleyn sign $\kappa$ and $\phi$ shown in Figure~\ref{fig:kasmat}. The Kasteleyn matrix is (with vertices ordered from left to right)
\[
	K_{(G^\square, \wt^\square,\kappa)}(z,w)=
\begin{bmatrix}
c_2 & s_2&0&\frac z w\\
s_2 &-c_2&1&0\\
0&w&-s_1&c_1\\
 \frac 1 z&0&c_1&s_1
\end{bmatrix}	,
\]
and the characteristic polynomial is 
\begin{align*}
	P_{(G^\square, \wt^\square,\kappa)}(z,w)=1+c_{1}^{2} c_{2}^{2} + c_{2}^{2} s_{1}^{2} + c_{1}^{2} s_{2}^{2} + s_{1}^{2} s_{2}^{2} - c_{2} s_{1}w -  c_{1} s_{2}z - \frac{ c_{2} s_{1}}{w} - \frac{c_{1} s_{2}}{z}.	
\end{align*}

The adjugate matrix $Q_{(G^\square, \wt^\square,\kappa)}(z,w)$ is 
\[
\begin{bmatrix}
	c_1^2 c_2+c_2 s_1^2-s_1 w & c_1^2 s_2-c_1 z+s_1^2 s_2 & s_1
	s_2-\frac{c_1 c_2 z}{w} & -c_1 s_2-\frac{c_2 s_1 z}{w}+z \\
	c_1^2 s_2-\frac{c_1}{z}+s_1^2 s_2 & -c_1^2 c_2-c_2 s_1^2+\frac{s_1}{w} &
	-\frac{c_1 s_2 z}{w}-c_2 s_1+\frac{1}{w} & c_1 c_2-\frac{s_1 s_2 z}{w} \\
	s_1 s_2 w-\frac{c_1 c_2}{z} & -\frac{c_1 s_2}{z}-c_2 s_1 w+1 & -c_2^2
	s_1+\frac{c_2}{w}-s_1 s_2^2 & c_1 c_2^2+c_1 s_2^2-s_2 z \\
	-c_1 s_2 w-\frac{c_2 s_1}{z}+\frac{w}{z} & c_1 c_2 w-\frac{s_1 s_2}{z} & c_1
	c_2^2+c_1 s_2^2-\frac{s_2}{z} & c_2^2 s_1-c_2 w+s_1 s_2^2 
\end{bmatrix}.
\]
Setting the $\bw$-column and $\b$-row respectively equal to $0$ (with $\bw$ and $\b$ as shown in Figure~\ref{fig:kasmat}), we get 
\[
D_{\bw} = \left(\frac{c_1}{s_2(c_1^2+s_1^2)},\frac{s_1}{c_2(c_1^2+s_1^2)} \right) \text{ and }D_\b= \left(\frac{s_2}{c_1(c_2^2+s_2^2)},\frac{c_2}{s_1(c_2^2+s_2^2)}\right).
\]
Note that $P_{(G^\square, \wt^\square,\kappa)}(z,w)$ is invariant under the involution $\sigma:(z,w) \mapsto (z^{-1},w^{-1})$. Using $c_1^2+s_1^2=c_2^2+s_2^2=1$, we also see that $D_{\b}=\sigma(D_\bw)$. Both of these properties are true for all spectral curves and standard divisors that arise from the Ising model and characterize them as we will prove in Theorem~\ref{thm:main}.
\end{example}

A \textit{spectral data} associated to $\Gamma$ is a triple $(C,D,\nu_\Gamma)$ where
\begin{enumerate}
	\item $C$ is a Harnack curve with Newton polygon $N_\Gamma$,
	\item $D$ is a standard divisor on $C$,
	\item $\nu_{\Gamma} : Z_{\Gamma} \ra \bigsqcup_{S \text{ side of }N} C \cap D_S$ is a bijection between zig-zag paths and points at infinity such that for every side $S$ of $N_\Gamma$, $\nu_\Gamma(Z_{\Gamma,S}) = C \cap D_S$.
\end{enumerate} 
Let $\mathcal S_\Gamma$ denote the set of spectral data associated to $\Gamma$.

The \textit{spectral transform} is the map 
\[
\lambda_{(\Gamma,\mathrm{v})}:\mathcal X_\Gamma(\Rpos) \times H^1(\T,\{\pm 1\}) \ra \mathcal S_\Gamma	
\]
sending $([\wt],[\kappa])$ to the spectral data $(C,D_{\mathrm v},\nu_\Gamma)$ where
\begin{enumerate} 
	\item $C$ is the spectral curve,
	\item $D_{\mathrm v}$ is the divisor as in (\ref{eq:div}),
	\item $\nu_{\Gamma}$ is the bijection in (\ref{eq:bij}).
\end{enumerate} 
Although the Kasteleyn matrix depends on the choice of cochains $\wt$ and $\kappa$ representing $[\wt]$ and $[\kappa]$, the spectral data only depends on the cohomology classes $[\wt]$ and $[\kappa]$. 

\begin{theorem}[\cite{KO}]\label{thm:komain} The spectral transform is a bijection.
\end{theorem}

\begin{figure}
	\begin{tikzpicture}[scale=0.6]
		\def\ep{0.3}
		
		% \node[](no) at (0,-3){(M1) The spider move.};
		\begin{scope}[shift={(3,0)},rotate=45]
			\def\r{2};
			% \draw[gray,dashed,-] (0,0) circle (\r cm);
			\fill[black!5] (0,0) circle (1.05*\r cm);
			\coordinate[wvert] (n1) at (0:\r);
			\coordinate[wvert] (n2) at (0+90:\r);
			\coordinate[wvert] (n3) at (0+180:\r);
			\coordinate[wvert] (n4) at (0+270:\r);
			
			\coordinate[bvert] (b1) at (0:0.5*\r);
			\coordinate[bvert] (b2) at (0+180:0.5*\r);
			
			\draw[-]
			(n1) -- node[above]{$1$} (b1) -- node[above]{$\kappa_2 $} (n2)-- node[left]{$\kappa_1 $} (b2)-- node[below]{$\kappa_4 $}(n4) -- node[right]{$\kappa_3 $} (b1)
			(b2)--node[below]{$1$}(n3)
			;
		 
		\end{scope}
		
		\begin{scope}[shift={(-4,0)},rotate=90+45]
			\def\r{2};
			% \draw[gray,dashed,-] (0,0) circle (\r cm);
			\fill[black!5] (0,0) circle (1.05*\r cm);
			\coordinate[wvert] (n1) at (0:\r);
			\coordinate[wvert] (n2) at (0+90:\r);
			\coordinate[wvert] (n3) at (0+180:\r);
			\coordinate[wvert] (n4) at (0+270:\r);
			
			\coordinate[bvert] (b1) at (0:0.5*\r);
			\coordinate[bvert] (b2) at (0+180:0.5*\r);
			
			\draw[-]
			(n1) -- node[above]{$1$} (b1) -- node[left]{$\kappa_1 $} (n2)-- node[below]{$\kappa_4 $} (b2)-- node[right]{$\kappa_3 $}(n4) -- node[above]{$\kappa_2 $} (b1)
			(b2)--node[below]{$1$}(n3)
			;
			\coordinate[] (t1) at (15:\r);
			\coordinate[] (t2) at (120-45:\r);
			\coordinate[] (t3) at (150-45:\r);
			\coordinate[] (t4) at (210-45:\r);
			\coordinate[] (t5) at (240-45:\r);
			\coordinate[] (t6) at (300-45:\r);
			\coordinate[] (t7) at (330-45:\r);
			\coordinate[] (t8) at (30-45:\r);

		  \end{scope}

		  \draw[] (6,-2) -- (6,2);

\node[](no) at (-0.5,0) {$\longleftrightarrow$};		

		\begin{scope}[shift={(12-2-1,0)}
			,rotate=-45
			]
			\def\r{2};
			\fill[black!5] (0,0) circle (\r cm);

			\coordinate[] (t1) at (15:\r);
			\coordinate[] (t2) at (120-45:\r);
			\coordinate[] (t3) at (150-45:\r);
			\coordinate[] (t4) at (210-45:\r);
			\coordinate[] (t5) at (240-45:\r);
			\coordinate[] (t6) at (300-45:\r);
			\coordinate[] (t7) at (330-45:\r);
			\coordinate[] (t8) at (30-45:\r);
			
			\coordinate[] (m1) at (0:1);
			\coordinate[] (m2) at (0+90:1);
			\coordinate[] (m3) at (0+180:1);
			\coordinate[] (m4) at (0+270:1);

			\coordinate[wvert] (n1) at (0:\r);
			\coordinate[wvert] (n2) at (0+90:\r);
			\coordinate[wvert] (n3) at (0+180:\r);
			\coordinate[wvert] (n4) at (0+270:\r);

			\coordinate[bvert] (n5) at (0,0);
			\draw[] (n5) edge (n2) edge (n3) edge (n4) edge (n1);
			\end{scope}
			\node[](no) at (13-0.5,0) {$\longleftrightarrow$};
				\begin{scope}[shift={(20-2-2,0)},rotate=-45]
			\def\r{2};
			\fill[black!5] (0,0) circle (\r cm);
			\coordinate[wvert] (n1) at (0:\r);
			\coordinate[wvert] (n2) at (0+90:\r);
			\coordinate[wvert] (n3) at (0+180:\r);
			\coordinate[wvert] (n4) at (0+270:\r);
			\coordinate[wvert] (n5) at (0,0);
			
			\coordinate[bvert] (b1) at (-0.5,0.5);
			\coordinate[bvert] (b2) at (0.5,-0.5);
			
			\draw[-]
							(n1)--(b2)--(n4)
							(n2)--(b1)--(n3)
							(b1)--node[left]{$ 1$}(n5)--node[left]{$-1$}(b2)
							;

			\end{scope}   

	  \end{tikzpicture}
	  \caption{Transformation of $\kappa$ under moves. } \label{fig:kas}
\end{figure}
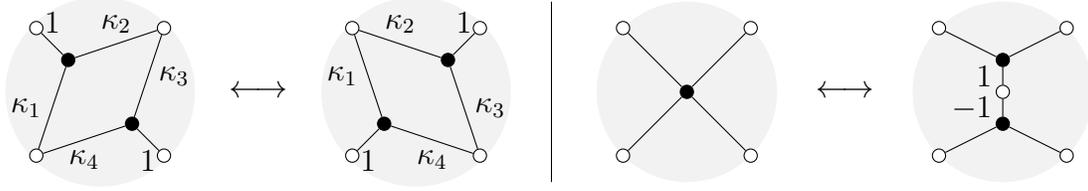

Suppose $\Gamma \stackrel{s}{\rightsquigarrow} \Gamma'$ is a move. By a gauge transformation, we can assume that $\wt$ and $\kappa$ are as shown on one side of Figure~\ref{fig:kas}. Then $\mu_s([\kappa])$ is as shown on the other side. 

Finally, we need a construction of Fock \cite{Fock} called the \textit{discrete Abel map}. Recall that $\pi:\R^2 \ra \T$ is the universal cover. Let $\tilde \Gamma := \pi^{-1}(\Gamma)$ denote the bi-periodic graph in the plane associated to $\Gamma$. The discrete Abel map $\mathbf{d}_{\tilde \Gamma}$ associates to each vertex of $\tilde \Gamma$ a divisor at infinity of $C$ and is defined as follows:
\begin{enumerate}
	\item Let $\mathbf{d}_{\tilde \Gamma}({\bw})=0$ where $\mathrm{w}$ is a fixed white vertex (this is a normalization).
	\item For any edge $\{\b,\bw\}$ contained in the fundamental rectangle with zig-zag paths $\alpha,\beta$ passing through it, \[\mathbf{d}_{\tilde \Gamma}(\b)-\mathbf{d}_{\tilde \Gamma}(\bw) =\nu_\Gamma(\alpha)+\nu_\Gamma(\beta).\]
\end{enumerate}
Then $\mathbf{d}_{\tilde \Gamma}({\mathrm v} + (i,j) ) = \mathbf{d}_{\tilde \Gamma}({\mathrm v})+ \operatorname{div} z^i w^j $ so $\mathbf{d}_{\tilde \Gamma}$ is not well-defined on $\Gamma$ (different lifts of a vertex have different values). However, since $z^i w^j$ is a rational function, the linear equivalence class of $\mathbf{d}_{\tilde \Gamma}(\mathrm{v})$ is well-defined for every vertex $\mathrm v$ of $\Gamma$ independent of the choice of lift of $\mathrm v$ to $\tilde \Gamma$; we denote this linear equivalence class by $\mathbf{d}_\Gamma({\mathrm v})$. 

A move $\Gamma \stackrel{s}{\rightsquigarrow} \Gamma'$ induces a discrete Abel map $\mathbf{d}_{\Gamma'}$ as follows: There is a unique way to define $\mathbf{d}_{\Gamma'}$ on $\Gamma'$ so that $\mathbf{d}_{\Gamma'}(\mathrm v) = \mathbf{d}_{\Gamma}(\mathrm v)$ for vertices $\mathrm v$ that are in both $\Gamma$ and $\Gamma'$; essentially we want the two discrete Abel maps to be consistently normalized.

The following theorem summarizes how the spectral data changes upon doing moves and changing the vertex $\rm v$.

\begin{theorem}\label{thm:mutation_spec} 
	Suppose $\Gamma \stackrel{s}{\rightsquigarrow} \Gamma'$ is a sequence of moves. Then the following diagram commutes
		\[
			\begin{tikzcd}
				\mathcal X_{\Gamma}(\Rpos)\times H^1(\T,\{\pm 1\}) \arrow[r,"\mu_s"]\arrow[d,"{\lambda_{(\Gamma,{\rm v})}}"'] & \mathcal X_{\Gamma'}(\Rpos)\times H^1(\T,\{\pm 1\}) \arrow[d,"{\lambda_{(\Gamma',{\rm v}')}}"]\\
				\mathcal S_{\Gamma} \arrow[r] & \mathcal S_{\Gamma'}
			\end{tikzcd},
			\]
		where the map $\mathcal S_\Gamma \ra \mathcal S_{\Gamma'}$ is given by $(C,D_{\rm v},\nu_{\Gamma}) \mapsto (C,D_{{\rm v}'}, \nu_{\Gamma'})$ where
\begin{enumerate}
	\item $D_{\mathrm{v}'}$ is the unique degree $g$ effective divisor satisfying (where $\sim$ denotes linear equivalence of divisors)
	\begin{enumerate}
		\item $D_{\rm v} + \mathbf{d}_\Gamma(\rm v) \sim D_{{\rm v}'} + \mathbf{d}_{\Gamma'}({\rm v}')$ if $\mathrm{v}$ and $\mathrm{v}'$ have the same color.
		\item $D_{\rm v} + D_{\mathrm{v}'} \sim K_C + \mathbf{d}_{\Gamma}({\mathrm v}) - \mathbf{d}_{\Gamma'}(\mathrm v')$ if $\mathrm v$ is black and $\mathrm v'$ is white, where $K_C$ is the canonical divisor of $C$.
	\end{enumerate}
\item	$
		\nu_{\Gamma'}(s_*(\alpha)) := \nu_{\Gamma}(\alpha).	
		$
\end{enumerate}
\end{theorem}

Theorem~\ref{thm:mutation_spec}(1)(a) is due to Fock \cite{Fock} and Theorem~\ref{thm:mutation_spec}(1)(b) is well known; see for example~\cite[Corollary 6.2]{GGK}~or~\cite[Lemma 32]{BDdT}.

\begin{comment}
Finally, we need the following result. 
\begin{proposition}[{\cite[Corollary 6.2]{GGK}}] \label{prop:ggk}
	%Suppose $\b$ is a black vertex and $\bw$ is a white vertex of $\Gamma$. Then  
	%\[
	%\div \Omega_{\b \bw}(z,w) = D_{\b}+D_{\bw} - d(\b)+d(\bw).	
	%\]
	Suppose $e=\{ \b,\bw\}$ is an edge of $\Gamma$ and suppose $\alpha,\beta$ are the two zig-zag paths that contain $e$. Then
	\[
	D_{\b} + D_{\bw} \sim K_C + \nu_{\Gamma}(\alpha)+\nu_{\Gamma}(\beta),	
	\]
	where $K_C$ is the canonical divisor of $C$ and $\sim$ denotes linear equivalence of divisors.
\end{proposition}
\end{comment}
\section{Color change}

Bipartite torus graphs have a global transformation which will play an important role in this paper. Let $(\overline\Gamma, [\overline \wt])$ be the weighted bipartite graph obtained from $(\Gamma,[\wt])$ by changing the colors of all the vertices and keeping the weights of all edges the same. For a vertex ${\rm v}$ in $\Gamma$ let $\overline{\rm v}$ denote the corresponding oppositely colored vertex in $\overline \Gamma$. So if $\wt$ is an edge weight representing $[\wt]$ and $e =\{ \b,\bw\}$ is an edge of $\Gamma$, then 
\[
	\overline\wt(\{\overline \bw,\overline \b\})=\wt(\{\b,\bw\})
\]
is an edge weight representing $[\overline \wt]$.

For a cycle $\gamma$ in $\Gamma$, let $\overline \gamma$ denote the corresponding cycle in $\overline \Gamma$. Then $\alpha \mapsto -\overline \alpha$ is a bijection between $Z_\Gamma$ and $Z_{\overline \Gamma}$ so $N_{\overline \Gamma}=-N_\Gamma$.  

If $\kappa$ is a Kasteleyn sign on $\Gamma$, then $\overline \kappa$ is a Kasteleyn sign on $\overline \Gamma$. The Kasteleyn matrices of $\Gamma$ and $\overline \Gamma$ are related by
\begin{equation} \label{eq:sym:kas}
	K_{(\overline \Gamma, \overline \wt, \overline \kappa)}(z,w)_{\overline \b \overline\bw } = K_{(\Gamma, \wt, \kappa)}(z^{-1},w^{-1})_{\bw \b}^T,
\end{equation}
so we get 
\[
	P_{(\overline \Gamma, \overline \wt, \overline \kappa)}(z,w)	= P_{(\Gamma, \wt, \kappa)}(z^{-1},w^{-1}).
\]
Let $\sigma: (\C^\times)^2 \ra (\C^\times)^2$ denote the involution $(z,w) \mapsto (z^{-1},w^{-1})$. Recall that $Q_{(\Gamma, \wt, \kappa)}(z,w)$ denotes the cofactor matrix of $K_{(\Gamma, \wt, \kappa)}(z,w)$ and that for a vertex $\rm v$ of $\Gamma$, the divisor $D_{\rm v}$ is defined by the vanishing of the row or column of $Q_{(\Gamma, \wt, \kappa)}(z,w)$ corresponding to $\rm v$. From (\ref{eq:sym:kas}), we get 
\begin{equation*} 
	Q_{(\overline \Gamma, \overline \wt, \overline \kappa)}(z,w)_{\overline\bw  \overline \b } = Q_{(\Gamma, \wt, \kappa)}(z^{-1},w^{-1})_{\b \bw }^T,
\end{equation*}
which implies that $D_{\overline {\rm v}} = \sigma(D_{\rm v})$.

Since $X_{\overline \alpha}([\overline \wt]) = X_{\alpha}([\wt])^{-1}$, multiplying (\ref{eq:bij}) by $-1$ we get
\begin{equation}\label{eq:nus}
	\nu_{\overline \Gamma}(\overline \alpha) = \sigma(\nu_\Gamma(\alpha)).
\end{equation}

We have shown:
\begin{proposition} \label{prop:colorchange}
	The following diagram commutes: 
	\[
\begin{tikzcd}
	\mathcal X_{\Gamma}(\Rpos)\times H^1(\T,\{\pm 1\}) \arrow[r]\arrow[d,"\lambda_{(\Gamma,\mathrm{v})}"'] &\mathcal X_{\overline \Gamma}(\Rpos) \times H^1(\T,\{\pm 1\}) \arrow[d,"{\lambda_{(\overline\Gamma,\overline{\mathrm{v}})}}"]\\
	\mathcal S_\Gamma \arrow[r] & \mathcal S_{\overline \Gamma}
\end{tikzcd}
\]
where the top map is $([\wt],[\kappa]) \mapsto ([\overline \wt],[\overline \kappa])$ and the bottom map is 
$(C,D,\nu_\Gamma) \mapsto (\sigma(C), \sigma(D), \nu_{\overline \Gamma})$ with $\nu_{\overline \Gamma}$ defined as in (\ref{eq:nus}).
\end{proposition}

\section{A characterization of $\mathcal I_G(\Rpos) \subset \mathcal X_{G^\square}(\Rpos)$}\label{sec:car_weight}

If we apply square moves at all the square faces of $G^\square$, then the resulting graph is $\overline{G^\square}$. Let $s_*:H_1(G^\square,\Z) \ra H_1(\overline{G^\square},\Z)$ denote the induced map of cycles and $\mu:\mathcal X_{G^\square}(\Rpos) \ra \mathcal X_{\overline{G^\square}}(\Rpos)$ be induced map of weights. Therefore, starting with $(G^\square,[\wt])$, there are two ways to get weights on $\overline{G^\square}$: $[\overline \wt]$ and $\mu([\wt])$. The main result of this section is:
\begin{theorem} \label{thm:main1}
	The subset $\mathcal I_{G}(\Rpos) \subset \mathcal X_{G^\square}(\Rpos)$ is the set of weights such that $\mu([\wt]) = [\overline{\wt}]$. 
\end{theorem}
The proof of Theorem~\ref{thm:main1} relies on the following lemma.
\begin{lemma} \label{lem::faces_determine}
	Suppose $[\wt] \in \mathcal X_{G^\square}(\Rpos)$ is such that 
	$
	\mu([\wt])=[\overline{\wt}].	
	$
	Then $[\wt]$ is determined by $(X_f([\wt]))_{\text{square faces }f}$.
\end{lemma}
\begin{proof}

Let $\gamma$ be a cycle in $G^\square$. Then we have $X_{\overline \gamma} ([\overline{\wt}]) = X_{\gamma}([\wt])^{-1}$. On the other hand, by (\ref{eq:fmut}), we have
\begin{equation}\label{eq:sgamam}
	\overline \gamma = s_*(\gamma) + \sum_{\text{square faces }f} a_f s_*(f),
\end{equation}
for some $a_f \in \Z$. Using (\ref{eq:mu}), we get
\[
 X_{ \overline \gamma} (\mu([\wt]))= X_\gamma([\wt]) \prod_{\text{square faces }f} X_f([\wt])^{a_f}(1+X_f([\wt]))^{-\langle \gamma,f \rangle} .	
\] 
Since $\mu([\wt])=[\overline{\wt}]$, we have
\begin{equation}\label{eq:Xpos}
	X_{\gamma}([\wt])^{2}=\prod_{\text{square faces }f} X_f([\wt])^{-a_f}(1+X_f([\wt]))^{\langle \gamma,f \rangle}.
\end{equation}

Therefore, every $X_{\gamma}^2$ is a function of $(X_f)_{\text{square faces }f}$. For positive weights, there is a natural square root, so $X_\gamma$ is determined by $(X_f)_{\text{square faces }f}$.
\end{proof}

\begin{figure}
	\begin{tikzpicture}[scale=0.5]

		 \begin{scope}[shift={(0,0)},rotate=0]
\def\ep{0.2}
			% \draw[gray,dashed,-] (0,0) circle (\r cm);
			\fill[black!5] (0,0) rectangle (8,8);
			\draw[gray,dashed] (0,0) rectangle (8,8);

			\coordinate[bvert] (b1) at (1,5);
			\coordinate[bvert] (b2) at (1+2,5+2);
			\coordinate[bvert] (b3) at (5,3);
			\coordinate[bvert] (b4) at (7,1);
  
			\coordinate[wvert] (w1) at (1,7);
			\coordinate[wvert] (w2) at (1+2,5);
			\coordinate[wvert] (w3) at (5,1);
			\coordinate[wvert] (w4) at (7,3);
  
			\draw[] (b1) -- (0,4) 
			(b1) -- (w2)
			(b1) --(w1);
  
			\draw[] (b2) -- (w1)
			(b2) -- (4,8)
			(b2) -- (w2);
  
			\draw[] (b3) -- (w3)
			(b3) -- (w2)
			(b3) -- (w4) ;
  
			\draw[] (b4) -- (w3)
			(b4) -- (8,0)
			(b4) -- (w4) ;
  
			\draw[] (w4) -- (8,4)
			(0,8) -- (w1)
			(4,0) -- (w3);

			\draw[red,->,line width=\lw,rounded corners=\rc] (0,4) -- (1+\ep,5-\ep) -- (3-\ep,5-\ep) -- (5+\ep,3+\ep) --(7-\ep,3+\ep)-- (8,4);

			\draw[blue,->,line width=\lw,rounded corners=\rc] (4,0) -- (5-\ep,1+\ep) -- (5-\ep,3-\ep) -- (3+\ep,5+\ep) --(3+\ep,7-\ep)-- (4,8); 

\node[](f1) at (2,6) {$f_2$};
\node[](f2) at (6,2) {$f_1$};
\node[](f3) at (2,2) {$f_3$};
\node[](f4) at (6,6) {$f_4$};

\node[](a) at (4,8.5) {$b$};
\node[](b) at (8.5,4) {$a$};

			\node[](no) at (4,-1) {(a)};

		   \end{scope}

		   \begin{scope}[shift={(10,0)},rotate=0]
			\def\ep{0.2}
						% \draw[gray,dashed,-] (0,0) circle (\r cm);
						\fill[black!5] (0,0) rectangle (8,8);
						\draw[gray,dashed] (0,0) rectangle (8,8);

						\coordinate[wvert] (b1) at (1,5);
						\coordinate[wvert] (b2) at (1+2,5+2);
						\coordinate[wvert] (b3) at (5,3);
						\coordinate[wvert] (b4) at (7,1);
			  
						\coordinate[bvert] (w1) at (1,7);
						\coordinate[bvert] (w2) at (1+2,5);
						\coordinate[bvert] (w3) at (5,1);
						\coordinate[bvert] (w4) at (7,3);
			  
						\draw[] (b1) -- (0,4) 
						(b1) -- (w2)
						(b1) --(w1);
			  
						\draw[] (b2) -- (w1)
						(b2) -- (4,8)
						(b2) -- (w2);
			  
						\draw[] (b3) -- (w3)
						(b3) -- (w2)
						(b3) -- (w4) ;
			  
						\draw[] (b4) -- (w3)
						(b4) -- (8,0)
						(b4) -- (w4) ;
			  
						\draw[] (w4) -- (8,4)
						(0,8) -- (w1)
						(4,0) -- (w3);
			
						\draw[red,->,line width=\lw,rounded corners=\rc] (0,4) -- (1+\ep,5-\ep) -- (1+\ep,7-\ep)--(3-\ep,7-\ep) -- (3-\ep,5-\ep) -- (5+\ep,3+\ep) --(7-\ep,3+\ep)-- (8,4);

						\draw[blue,->,line width=\lw,rounded corners=\rc] (4,0) --(5-\ep,1+\ep) -- (7-\ep,1+\ep)--(7-\ep,3-\ep)--  (5-\ep,3-\ep) -- (3+\ep,5+\ep) --(3+\ep,7-\ep)-- (4,8); 
			
			\node[](f1) at (2,6) {$\overline{f_2}$};
			\node[](f2) at (6,2) {$\overline{f_1}$};
			\node[](f3) at (2,2) {$\overline{f_3}$};
			\node[](f4) at (6,6) {$\overline{f_4}$};
			
			\node[](a) at (4,8.5) {$s_*(b)$};
			\node[](b) at (9,4) {$s_*(a)$};

						\node[](no) at (4,-1) {(b)};

					   \end{scope}

	  \end{tikzpicture}
	  \caption{(a) A basis for $H_1(G^\square,\Z)$ and (b) the cycles corresponding to $a$ and $b$ in $\overline{G^\square}$. From (b), we see that $\overline a=s_*(a)-\overline{f_2}= s_*(a) + s_*(f_2)$ and similarly the equation for $\overline b$. 
	  } \label{fig:ising_torus_check}
\end{figure}
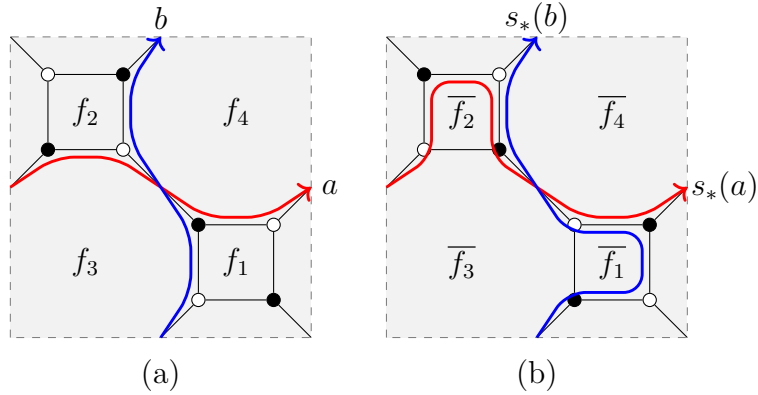

\begin{example}
	Consider the basis $\{f_1,f_2,f_3,a,b\}$ for $H_1(G^\square,\Z)$ shown in Figure~\ref{fig:ising_torus_check}(a). Then the equations (\ref{eq:sgamam}) are (see Figure~\ref{fig:ising_torus_check}(b))
\begin{align*}
	\overline{f_1} &= s_*(f_1) -2 s_*(f_1), & 
	\overline{f_2} &= s_*(f_2)  -2 s_*(f_2),\\
	\overline{f_3} &= s_*(f_3) +2 s_*(f_1) + 2 s_*(f_2), &
	\overline{a} &= s_*(a) + s_*(f_2),\\
	\overline{b} &= s_*(b) - s_*(f_1),
\end{align*}
and the equations (\ref{eq:Xpos}) are 
\begin{align}\label{eq:X2}
	X_{f_1}^2 &=X_{f_1}^2, & X_{f_2}^2&=X_{f_2}^2,\nonumber\\
	X_{f_3}^2&= X_{f_1}^{-2} X_{f_2}^{-2} (1+X_{f_1})^{2}(1+X_{f_2})^{2}, & X_a^2&=X_{f_2}(1+X_{f_1})^{-1}(1+X_{f_2})^{-1},\nonumber\\
	X_b^2&=X_{f_1}^{-1}(1+X_{f_1})(1+X_{f_2}).
\end{align}
Theorem~\ref{thm:main1} says that the subset $\mathcal I_{G}(\Rpos) \subset \mathcal X_{G^\square}(\Rpos)$ is defined by these five equations (two of which are trivial). Using the weights in Figure~\ref{fig:isingintro}(c), we compute
\begin{align*}
	X_{f_1} =\frac{s_1^2}{c_1^2},\quad X_{f_2}=\frac{s_2^2}{c_2^2},\quad
	X_{f_3}= \frac{1}{s_1^2s_2^2},\quad  X_a=c_1 s_2,\quad
	X_b=\frac{1}{c_2s_1}.
\end{align*}
Plugging into (\ref{eq:X2}), we see that all the equations are satisfied.
\end{example}

\begin{proof}[Proof of Theorem~\ref{thm:main1}]
	Let $[\wt]$ be a weight such that $\mu([\wt]) = [\overline{\wt}]$. For an edge $e \in E(G)$, let $f$ denote the corresponding square in $G^\square$. Consider the Ising weight $\wt^\square$ on $G^\square$ with $s_e := \sqrt{ \frac{X_f}{1+X_f}}$ and $c_e := \sqrt{\frac{1}{1+X_f}}$. Then $[\wt^\square]$ also satisfies 
	\[
	\mu([\wt^\square]) = [\overline{\wt^\square}].	
	\]
Since $X_f([\wt^\square])=X_f([\wt])$ for all square faces $f$, $[\wt^\square]=[\wt]$ by Lemma~\ref{lem::faces_determine}.

\end{proof}

\section{Ising spectral data}\label{sec:isingspecdata}

The goal of this section is to describe the subset of spectral data that corresponds to Ising models, i.e., the image of 
\[
	\mathcal I_G(\Rpos) \hookrightarrow \mathcal X_{G^\square}(\Rpos) \ra \mathcal S_{G^\square}.	
\]

	\begin{theorem} \label{thm:main}
		Let $G$ be a minimal graph in $\T$, $G^\square$ the corresponding bipartite graph, ${\bw}$ a white vertex of $G^\square$ and $[\kappa]$ a Kasteleyn sign on $G^\square$. Suppose ${\b}$ denotes the black vertex incident to $\bw$ that is not a part of the square containing $\bw$ and let $\alpha,\overline \alpha$ denote the zig-zag paths that contain the edge $\{\b,\bw\}$. Let $[\wt] \in \mathcal X_{G^\square}(\Rpos)$ and let $(C,D_{{\bw}},\nu_{G^\square}) = \lambda_{(G^\square,{\bw})} ([\wt],[\kappa])$. Then $ [\wt] \in \mathcal I_G(\Rpos)$ if and only if 
		\begin{enumerate}
			\item $C$ is invariant under the involution $\sigma:(z,w) \ra (z^{-1},w^{-1})$.
			\item $D_{{\bw}}$ is a standard divisor satisfying
			 \[
				D_{{\bw}} + \sigma(D_{{\bw}}) \sim K_C + \nu_{G^\square}(\alpha) +\nu_{G^\square}(\overline \alpha).
	\]
	Such divisors are in the Prym variety of $C$ (see e.g., \cite[Section 3]{MvM}). If $C$ is singular, we replace it with its desingularization $\tilde C$ as in Remark~\ref{rem:singular}.
			 \item For every zig-zag path $\alpha$ of $G^\square$ we have
			 \[
				\nu_{G^\square}(\overline \alpha) = \sigma(\nu_{G^\square}(\alpha)). 
				\]
		\end{enumerate}
	\end{theorem}

\begin{proof}
	Since $\mathbf{d}_{G^\square}(\b)-\mathbf{d}_{G^\square}(\bw) = \nu_{G^\square}(\alpha) +\nu_{G^\square}(\overline \alpha)$, by Theorem~\ref{thm:mutation_spec}(1)(b) the condition (2) is equivalent to 
	\[
	(2')~D_{{\bw}} = \sigma(D_{{\b}}).	
	\]
	Now we have
	\begin{align*}
		&[\wt] \in \mathcal I_{G}(\Rpos) &&\\& \Longleftrightarrow [\overline{\wt}] = \mu([\wt]) && \text{(Theorem~\ref{thm:main1})}\\
		& \Longleftrightarrow \lambda_{(\overline\Gamma,\overline{{\b}})}([\overline{\wt}],[\overline\kappa])= \lambda_{(\overline\Gamma,\overline{{\b}})}(\mu([\wt]),\mu([\kappa]))&& \text{(Theorem~\ref{thm:komain} and $\mu([\kappa])=[\overline \kappa]$ (Figure~\ref{fig:kas}))}\\
		& \Longleftrightarrow (\sigma(C), \sigma(D_{{\b}}), \sigma(\nu_{G^\square} (\overline \cdot)) ) = 	(C, D_{{\bw}},\nu_{G^\square}(\cdot) ) && \text{(Proposition~\ref{prop:colorchange} and Theorem~\ref{thm:mutation_spec})}\\
		& \Longleftrightarrow \text{(1), ($2'$) and (3)}. 
	\end{align*}
\end{proof}

	\bibliographystyle{alpha}
	\bibliography{biblio}

\begin{thebibliography}{GLSBS22}

\bibitem[AGPR24]{AGPR}
Niklas Affolter, Max Glick, Pavlo Pylyavskyy, and Sanjay Ramassamy.
\newblock Vector-relation configurations and plabic graphs.
\newblock {\em Selecta Math. (N.S.)}, 30(1):Paper No. 9, 2024.

\bibitem[Bax78]{Baxter1}
R.~J. Baxter.
\newblock Solvable eight-vertex model on an arbitrary planar lattice.
\newblock {\em Philos. Trans. Roy. Soc. London Ser. A}, 289(1359):315--346, 1978.

\bibitem[Bax86]{Baxter2}
R.~J. Baxter.
\newblock Free-fermion, checkerboard and {${\bf Z}$}-invariant lattice models in statistical mechanics.
\newblock {\em Proc. Roy. Soc. London Ser. A}, 404(1826):1--33, 1986.

\bibitem[Bax89]{Baxter3}
Rodney~J. Baxter.
\newblock {\em Exactly solved models in statistical mechanics}.
\newblock Academic Press, Inc. [Harcourt Brace Jovanovich, Publishers], London, 1989.
\newblock Reprint of the 1982 original.

\bibitem[BB23]{BB}
Tomas Berggren and Alexei Borodin.
\newblock Geometry of the doubly periodic aztec dimer model, 2023.

\bibitem[BCdT23]{BDdT}
C\'{e}dric Boutillier, David Cimasoni, and B\'{e}atrice de~Tili\`ere.
\newblock Minimal bipartite dimers and higher genus {H}arnack curves.
\newblock {\em Probab. Math. Phys.}, 4(1):151--208, 2023.

\bibitem[BD23]{BDuits}
Alexei Borodin and Maurice Duits.
\newblock Biased {$2 \times 2$} periodic {A}ztec diamond and an elliptic curve.
\newblock {\em Probab. Theory Related Fields}, 187(1-2):259--315, 2023.

\bibitem[BdT11]{BdT}
C\'{e}dric Boutillier and B\'{e}atrice de~Tili\`ere.
\newblock The critical {$Z$}-invariant {I}sing model via dimers: locality property.
\newblock {\em Comm. Math. Phys.}, 301(2):473--516, 2011.

\bibitem[BdT14]{CdT2}
C\'{e}dric Boutillier and B\'{e}atrice de~Tili\`ere.
\newblock Height representation of {XOR}-{I}sing loops via bipartite dimers.
\newblock {\em Electron. J. Probab.}, 19:no. 80, 33, 2014.

\bibitem[BdTR19]{BdTR}
C\'{e}dric Boutillier, B\'{e}atrice de~Tili\`ere, and Kilian Raschel.
\newblock The {$Z$}-invariant {I}sing model via dimers.
\newblock {\em Probab. Theory Related Fields}, 174(1-2):235--305, 2019.

\bibitem[BGKT21]{BGKT}
Boris Bychkov, Vassily Gorbounov, Anton Kazakov, and Dmitry Talalaev.
\newblock Electrical networks, lagrangian grassmannians and symplectic groups, 2021.

\bibitem[CCK17]{CCK}
Dmitry Chelkak, David Cimasoni, and Adrien Kassel.
\newblock Revisiting the combinatorics of the 2{D} {I}sing model.
\newblock {\em Ann. Inst. Henri Poincar\'{e} D}, 4(3):309--385, 2017.

\bibitem[CDC13]{CDC}
David Cimasoni and Hugo Duminil-Copin.
\newblock The critical temperature for the {I}sing model on planar doubly periodic graphs.
\newblock {\em Electron. J. Probab.}, 18:no. 44, 18, 2013.

\bibitem[CGS21]{CGS}
Sunita Chepuri, Terrence George, and David~E Speyer.
\newblock Electrical networks and lagrangian grassmannians, 2021.

\bibitem[Che18]{Csur}
Dmitry Chelkak.
\newblock Planar {I}sing model at criticality: state-of-the-art and perspectives.
\newblock In {\em Proceedings of the {I}nternational {C}ongress of {M}athematicians---{R}io de {J}aneiro 2018. {V}ol. {IV}. {I}nvited lectures}, pages 2801--2828. World Sci. Publ., Hackensack, NJ, 2018.

\bibitem[Cim12]{Cim}
David Cimasoni.
\newblock The critical {I}sing model via {K}ac-{W}ard matrices.
\newblock {\em Comm. Math. Phys.}, 316(1):99--126, 2012.

\bibitem[CS11]{CS1}
Dmitry Chelkak and Stanislav Smirnov.
\newblock Discrete complex analysis on isoradial graphs.
\newblock {\em Adv. Math.}, 228(3):1590--1630, 2011.

\bibitem[CS12]{CS2}
Dmitry Chelkak and Stanislav Smirnov.
\newblock Universality in the 2{D} {I}sing model and conformal invariance of fermionic observables.
\newblock {\em Invent. Math.}, 189(3):515--580, 2012.

\bibitem[Dub11]{Dubedat}
Julien Dubédat.
\newblock Exact bosonization of the ising model, 2011.

\bibitem[Foc15]{Fock}
V.~V. Fock.
\newblock Inverse spectral problem for gk integrable system.
\newblock 2015.

\bibitem[FW70]{FW}
Chungpeng Fan and F.~Y. Wu.
\newblock General lattice model of phase transitions.
\newblock {\em Phys. Rev. B}, 2:723--733, Aug 1970.

\bibitem[Gal22]{Gal}
Pavel Galashin.
\newblock A formula for boundary correlations of the critical {I}sing model.
\newblock {\em Probab. Theory Related Fields}, 182(1-2):615--640, 2022.

\bibitem[Geo24]{G1}
Terrence George.
\newblock Spectra of {B}iperiodic {P}lanar {N}etworks.
\newblock {\em Comm. Math. Phys.}, 405(1):10, 2024.

\bibitem[GGK23]{GGK}
T.~George, A.~B. Goncharov, and R.~Kenyon.
\newblock The inverse spectral map for dimers.
\newblock {\em Math. Phys. Anal. Geom.}, 26(3):Paper No. 24, 51, 2023.

\bibitem[GK13]{GK}
Alexander~B. Goncharov and Richard Kenyon.
\newblock Dimers and cluster integrable systems.
\newblock {\em Ann. Sci. \'{E}c. Norm. Sup\'{e}r. (4)}, 46(5):747--813, 2013.

\bibitem[GLSBS22]{GLSS}
Pavel Galashin, Thomas Lam, Melissa Sherman-Bennett, and David Speyer.
\newblock Braid variety cluster structures, i: 3d plabic graphs, 2022.

\bibitem[GP20]{Galpy}
Pavel Galashin and Pavlo Pylyavskyy.
\newblock Ising model and the positive orthogonal {G}rassmannian.
\newblock {\em Duke Math. J.}, 169(10):1877--1942, 2020.

\bibitem[HWX14]{HW}
Yu-Tin Huang, Congkao Wen, and Dan Xie.
\newblock The positive orthogonal grassmannian and loop amplitudes of abjm.
\newblock {\em Journal of Physics A: Mathematical and Theoretical}, 47(47):474008, nov 2014.

\bibitem[Ken09]{Ken}
Richard Kenyon.
\newblock Lectures on dimers.
\newblock In {\em Statistical mechanics}, volume~16 of {\em IAS/Park City Math. Ser.}, pages 191--230. Amer. Math. Soc., Providence, RI, 2009.

\bibitem[KO06]{KO}
Richard Kenyon and Andrei Okounkov.
\newblock Planar dimers and {H}arnack curves.
\newblock {\em Duke Math. J.}, 131(3):499--524, 2006.

\bibitem[KOS06]{KOS}
Richard Kenyon, Andrei Okounkov, and Scott Sheffield.
\newblock Dimers and amoebae.
\newblock {\em Ann. of Math. (2)}, 163(3):1019--1056, 2006.

\bibitem[KP16]{KP}
Richard Kenyon and Robin Pemantle.
\newblock Double-dimers, the {I}sing model and the hexahedron recurrence.
\newblock {\em J. Combin. Theory Ser. A}, 137:27--63, 2016.

\bibitem[KS04]{KSh}
Richard~W. Kenyon and Scott Sheffield.
\newblock Dimers, tilings and trees.
\newblock {\em J. Combin. Theory Ser. B}, 92(2):295--317, 2004.

\bibitem[Lam18]{Lam}
Thomas Lam.
\newblock Electroid varieties and a compactification of the space of electrical networks.
\newblock {\em Adv. Math.}, 338:549--600, 2018.

\bibitem[Li14]{ZLi}
Zhongyang Li.
\newblock Spectral curves of periodic {F}isher graphs.
\newblock {\em J. Math. Phys.}, 55(12):123301, 25, 2014.

\bibitem[MRr01]{Mik}
Grigory Mikhalkin and Hans Rullg\aa~rd.
\newblock Amoebas of maximal area.
\newblock {\em Internat. Math. Res. Notices}, (9):441--451, 2001.

\bibitem[Pos06]{Post}
Alexander Postnikov.
\newblock Total positivity, grassmannians, and networks, 2006.

\bibitem[vMM79]{MvM}
Pierre van Moerbeke and David Mumford.
\newblock The spectrum of difference operators and algebraic curves.
\newblock {\em Acta Math.}, 143(1-2):93--154, 1979.

\end{thebibliography}
\end{document}